\documentclass[12pt, onecolumn]{IEEEtran}
\IEEEoverridecommandlockouts
\usepackage[utf8]{inputenc} 
\usepackage[T1]{fontenc}
\usepackage{url}
\usepackage{ifthen}
\usepackage{cite}
\usepackage[cmex10]{amsmath} 
\usepackage{amssymb}
\usepackage{graphicx}
\usepackage{epstopdf}
\usepackage{epsfig}
\usepackage{multirow}
\usepackage{amsthm}
\usepackage{comment}
\usepackage{caption}
\usepackage{enumitem}
\usepackage{amsmath}
\usepackage{color}
\usepackage{bbm}
\usepackage{algorithm}
\usepackage{algpseudocode}
\usepackage{csquotes}
\usepackage{subfiles}
\usepackage{caption, multirow, makecell}
\usepackage{array}
\usepackage{caption}
\usepackage{subcaption}
\usepackage{afterpage}
\usepackage{dblfloatfix}
\usepackage{hyperref}
\usepackage{booktabs}
\usepackage{tikz}
\usepackage{longtable}
\usepackage{blkarray}
\usepackage{adjustbox}
\newtheorem{thm}{Theorem}
\newtheorem{lem}{Lemma}

\newtheorem{cor}{Corollary}
\newtheorem{example}{Example}

\newtheorem{defn}{Definition}
\newtheorem{conj}{Conjecture}

\newtheorem{cons}{Construction}

\newtheorem{rem}{Remark}

\def\BibTeX{{\rm B\kern-.05em{\sc i\kern-.025em b}\kern-.08em
    T\kern-.1667em\lower.7ex\hbox{E}\kern-.125emX}}

\begin{document}

\title{Plotkin-like Bound and Explicit Function-Correcting Code Constructions for Lee Metric Channels}
\author{\IEEEauthorblockN{Hareesh K., Rashid Ummer N.T. and B. Sundar Rajan} \\
	\IEEEauthorblockA{\textit{Department of Electrical Communication Engineering} \\
		\textit{Indian Institute of Science,}
		Bangalore, India \\
		\{hareeshk,rashidummer,bsrajan\}@iisc.ac.in}
}

\maketitle

\begin{abstract}
	Function-Correcting Codes (FCCs) are a novel class of codes designed to protect function evaluations of messages against errors while minimizing redundancy. A theoretical framework for systematic FCCs to channels matched to the Lee metric has been studied recently, which introduced function-correcting Lee codes (FCLCs) and also derived upper and lower bounds on their optimal redundancy. In this paper, we first propose a Plotkin-like bound for irregular Lee-distance codes. We then construct explicit FCLCs for specific classes of functions, including the Lee weight, Lee weight distribution, modular sum and locally bounded function. For these functions, lower bounds on redundancy are obtained, and our constructions are shown to be optimal in certain cases. Finally, a comparative analysis with classical Lee error-correcting codes and codes correcting errors in function values demonstrates that FCLCs can significantly reduce redundancy while preserving function correctness.
\end{abstract}


\section{Introduction}

In conventional communication systems, a sender transmits a message to a receiver through an error-prone channel. Traditional error-correcting codes, paired with appropriate decoders, aim to recover the entire message, treating all parts as equally significant. However, in many practical scenarios, the receiver is primarily interested in a specific attribute of the message, i.e., the value of a function evaluated on the message-rather than only reconstructing the entire message. While recovering the entire message naturally allows the receiver to compute the desired function, this approach can be inefficient when the message is long and the image of the function is small. For instance, in IoT sensor networks, a receiver may only need the maximum temperature reading (a function) rather than all sensor measurements (the full message). To address this inefficiency, Lenz et al. \cite{Lenz} introduced FCCs. These codes are designed to protect the evaluation of functions from errors, offering reduced redundancy and improved transmission efficiency.

The framework of FCCs over binary symmetric channels was first introduced in the seminal work by Lenz et al. \cite{Lenz}. Systematic encoding is preferred because in applications such as distributed computing and archival storage, preserving the original data is often essential. Since the principal advantage of FCCs lies in their reduced redundancy, the key objective is to obtain the smallest amount of redundancy referred to as optimal redundancy, that allows the reliable recovery of the attribute. To achieve this, the authors introduced the concept of irregular-distance codes and established a fundamental connection between their shortest achievable lengths and the optimal redundancy of FCCs. Leveraging this connection, general upper and lower bounds on the optimal redundancy were derived in \cite{Lenz} for general functions, and subsequently applied these results to specific function classes.

\subsection{Motivation}
As emphasized in \cite{Lenz}, investigating FCCs under various channel models is a valuable direction of research, as it can both deepen our theoretical understanding of FCCs and broaden their applicability to practical scenarios. One such relevant channel model is the Lee metric channel, where errors are characterized not by arbitrary symbol changes, but by small numerical deviations \cite{Lee}. The Lee metric is a distance measure used in coding theory, particularly useful for codes over modulo $q$ integers  $(\mathbb{Z}_{q})$, where $q \ge 2$. It is especially relevant in communication systems where errors occur as small magnitude changes (e.g., $+1$ or $-1$) in transmitted symbols. In classical coding theory, the choice of a distance metric is closely related to the nature of the communication channel and is dependent on the decoding scheme used\cite{CW} as well. The Hamming metric is well-matched to the binary symmetric channel (BSC) for maximum likelihood decoding (MLD), where bit-flip errors occur independently with equal probability. However, in many practical scenarios, especially in non-binary systems and phase-modulated schemes, the Hamming metric may no longer reflect the true error characteristics of the channel. In such cases, alternative metrics like the Lee metric offer a more suitable model. A discrete,  memoryless, symmetric channel as shown in Fig.~\ref{Fig_Lee} is strictly matched to the Lee metric for MLD if the probability of symbol errors depends on their Lee distance from the transmitted symbol. In Fig.~\ref{Fig_Lee}, $M = \left[ \frac{q}{2} \right]$, where $[x]$ denotes the integer part of $x$. The conditional probability $\Pr(i|0) = p_i$ for $i = 0, 1, 2, \ldots, M$. The probability $\Pr(-i|0) = p_i$ holds for $i = 1, 2, \ldots, M$. Finally, $\Pr(i|j) = \Pr(k|0)$, where $k \equiv i-j \ (\mathrm{mod}\ q)$. For instance, if the alphabet size is $q$, the conditional probabilities $p_i$ and $p_{-i}$ represent the likelihoods of shifting by $i$  units in either direction. The symmetry of the channel ensures that the probability of an error depends only on the magnitude of the shift, not its direction.  

\begin{figure}[h!]
	\centering
	\begin{tikzpicture}[scale=1]
		% Origin node
		\node at (0,0) [circle, fill=black, inner sep=1pt] {};
		\node[left] at (-0.1,0) {0};
		
		% Positive branches (0, 1, 2)
		\foreach \i/\labelx/\probx in {
			0/0/p_0,
			1/1/p_1,
			2/2/p_2
		} {
			\draw[thick] (0,0) -- (2, \i*0.5);
			\node[circle, fill=black, inner sep=1pt] at (2, \i*0.5) {};
			\node[right] at (2.1, \i*0.5) {$\scriptstyle \labelx$};
			\node[right] at (3.4, \i*0.5) {$\scriptstyle \probx$};
		}
		
		% Positive dots after node 2 (aligned with probabilities)
		\node at (2.1, 1.2) {$\cdot$};
		\node at (2.1, 1.5) {$\cdot$};
		
		\node at (3.7, 1.2) {$\cdot$};
		\node at (3.7, 1.5) {$\cdot$};
		
		% Positive last node connected
		\draw[thick] (0,0) -- (2, 1.9);
		\node[circle, fill=black, inner sep=1pt] at (2, 1.9) {};
		\node[right] at (2.1, 1.9) {$\scriptstyle \left[ \frac{q}{2} \right]$};
		\node[right] at (3.4, 1.9) {$p_M$};
		
		% Negative branches (-1, -2)
		\foreach \i/\labelx/\probx in {
			1/-1/p_1,
			2/-2/p_2
		} {
			\draw[thick] (0,0) -- (2, -\i*0.5);
			\node[circle, fill=black, inner sep=1pt] at (2, -\i*0.5) {};
			\node[right] at (2.1, -\i*0.5) {$\scriptstyle \labelx$};
			\node[right] at (3.4, -\i*0.5) {$\scriptstyle \probx$};
		}
		
		% Negative dots after node -2 (aligned with probabilities)
		\node at (2.1, -1.2) {$\cdot$};
		\node at (2.1, -1.5) {$\cdot$};
		
		\node at (3.7, -1.2) {$\cdot$};
		\node at (3.7, -1.5) {$\cdot$};
		
		% Negative last node connected
		\draw[thick] (0,0) -- (2,-1.9);
		\node[circle, fill=black, inner sep=1pt] at (2,-1.9) {};
		\node[right] at (2.1,-1.9) {$\scriptstyle -\left[ \frac{q-1}{2} \right]$};
		\node[right] at (3.4,-1.9) {$p_M$};
		
	\end{tikzpicture}
	\caption{Conditional probabilities for a discrete, memoryless, symmetric channel matched to the Lee metric.}
	\label{Fig_Lee}
\end{figure}

The Lee metric is ideally suited for systems where data is represented over finite rings $(\mathbb{Z}_{q})$ and errors occur as small-magnitude deviations rather than complete symbol substitutions \cite{CW}. This makes it highly relevant for practical communication systems, including phase modulation schemes\cite{KN}, multi-level flash memories\cite{AB}, and certain classes of networked and distributed architectures. In particular, Lee metric-based codes have found applications in constrained and partial-response channels\cite{RS}, interleaving schemes\cite{MJ}, orthogonal frequency-division multiplexing (OFDM)\cite{KUS}, and multidimensional burst-error correction\cite{TY}. In such settings, the Hamming metric, which is commonly used to count the number of symbol mismatches, is less appropriate, as it does not account for the magnitude of symbol transitions. In contrast, the Lee metric is better suited, as it quantifies the exact distance between symbols from a finite ring such as $\mathbb{Z}_{q}$, particularly when errors correspond to small perturbations. Extending function-correcting codes to the Lee metric not only enhances the robustness of such systems but also opens up new avenues in algebraic coding theory by enabling code design over rings rather than fields. Consequently, the development of function-correcting codes under the Lee metric represents both a significant theoretical advancement and an important practical tool for next-generation communication and computation systems.

\subsection{Related Works}
Xia et al. \cite{QHB} extended the concept of FCCs to symbol-pair read channels, introducing Function-Correcting Symbol-Pair Codes (FCSPCs). In \cite{QHB}, the authors focused on specific classes of functions—such as pair-locally binary functions, pair-weight functions, and pair-weight distribution functions—and provided explicit constructions of FCSPCs for these cases. Premlal and Rajan in \cite{PR} established a lower bound on the redundancy of FCCs. Notably, when the function under consideration is bijective, FCCs reduce to classical error-correcting codes (ECCs), implying that this bound also applies to systematic ECCs. The tightness of this bound was demonstrated for a certain range of parameters. Additionally, the authors analyzed FCCs for linear functions and showed that the upper bound proposed by Lenz et al. is tight by constructing optimal codes for a class of such functions. In \cite{GZ},  Ge et al. focused on FCCs for two important function classes: the Hamming weight and Hamming weight distribution functions. The authors presented improved redundancy bounds and proposed optimal constructions that achieve the lower bound in these cases. A generalization of FCCs to b -symbol-pair read channels was introduced by Singh et al. in \cite{AY}, resulting in the notion of b -symbol-pair function-correcting codes. The authors established both lower and upper bounds on the optimal redundancy required for general functions in this setting. Sampath and Rajan in \cite{SR} investigated FCCs for linear functions in the context of b -symbol read channels and derived a Plotkin-like bound for such codes. The concept of locally $(\rho,\lambda)$-function-correcting codes was introduced by Rajput et al. in \cite{CR}, where an upper bound on the redundancy of such codes was derived. An upper bound on the redundancy of FCCs over finite fields was derived by Ly et al. in \cite{LE}, who conjectured that this bound holds universally across all finite fields. Later, in \cite{GAA} Gyanendra et. al extended function-correcting $b$-symbol codes for locally $(\lambda, \rho, b)$-functions and discussed the possible values of $\lambda$ and $\rho$ for which any function can be considered as locally $(\rho,\lambda)$-function in $b$-symbol metric. Recently, Liu et al. in \cite{LL} introduced a new class of FCCs known as Function-Correcting Codes with Homogeneous Distance (FCCHDs), deriving several bounds on the optimal redundancy for certain classes of functions. Very recently, a theoretical framework for systematic FCCs for channels matched to the Lee metric has been studied by Gyanendra et al. in \cite{GA}. Upper bounds and lower bounds on the optimal redundancy were also derived in \cite{GA}.  
\subsection{Contributions}
In \cite{GA}, the authors developed theoretical foundations of FCCs for Lee metric channels which introduce at most $t$ errors and derived bounds on optimal redundancy of FCLCs. These bounds were then simplified and applied to Lee weight function and Lee weight distribution function. However, \cite{GA} did not provide any constructions of FCLCs for these functions.  We construct explicit FCLCs for specific classes of functions, including the Lee weight, Lee weight distribution, and modular sum function. For these functions, lower bounds on redundancy are obtained, and our constructions are shown to be optimal in certain cases. Our results demonstrate that FCLCs can achieve significantly lower redundancy than both classical Lee error-correcting codes and codes that correct errors in function values, while still ensuring accurate function evaluation in the presence of errors.
The contributions of this paper are summarized as follows:
\begin{itemize}
	\item A Plotkin-like lower bound is presented for irregular-Lee-distance codes over $\mathbb{Z}_q$, the ring of integers modulo $q$.  
	\item Explicit constructions are provided for specific classes of functions in Lee metric, including:
	\begin{itemize}
		\item  Lee weight function,
		\item  Lee weight distribution function,
		\item  Modular sum function.
		\item  Locally bounded function.
	\end{itemize}
	\item Lower bounds on redundancy are derived for the Lee weight, Lee weight distribution and modular sum functions by simplifying the Plotkin-like bounds in certain cases.
	\item The redundancy of these functions is shown to be optimal for certain values of the alphabet size $q$, message length $k$, and number of errors $t$.
\end{itemize} 
\subsection{Organization}
The rest of the paper is organized as follows. In Section~\ref{preliminaries}, we review the basic concepts and definitions related to Lee metric codes, FCCs, and irregular-Lee-distance codes. In Section~\ref{General  Results on the Optimal Redundancy}, the connection between FCLCs and irregular-Lee-distance codes, and bounds on the optimal redundancy of FCLCs \cite{GA} are reviewed. A variant of Plotkin-like bound is also presented for irregular-Lee-distance codes in this section. In Section~\ref{Lee Functions}, we apply these general results to specific classes of functions, including the Lee weight function, Lee weight distribution function, modular sum function, and locally bounded function providing explicit constructions and corresponding redundancy values and bounds. Section~\ref{Redundancy Comparisons} presents a comparative analysis of FCLCs with classical Lee error-correcting codes and codes that correct errors in function values, demonstrating the redundancy gains achieved by FCLCs. Finally, Section~\ref{Conclusion} concludes the paper with a summary of key results and potential directions for future work. 

\noindent \textit{Notations:}
Let $\mathbb{N}_0$ denote the set of non-negative integers and $\mathbb{N}_{0}^{M\times M}$ represents the set of all  $M\times M$ matrices with non-negative integer entries. For a matrix $\boldsymbol{D}\in\mathbb{N}_{0}^{M\times M}$, we denote its $(i,j)-th$  entry by $[\boldsymbol{D}]_{ij}$. Given two vectors $\boldsymbol{\mathbf{x},\mathbf{y}}\in \mathbb{Z}_q^n$, $d_H(\boldsymbol{\mathbf{x},\mathbf{y}})$ denotes the Hamming distance between $\mathbf{x}$ and $\mathbf{y}$ and $d_L(\boldsymbol{\mathbf{x},\mathbf{y}})$ denotes the Lee distance between $\boldsymbol{\mathbf{x}}$ and $\boldsymbol{\mathbf{y}}$  for $\boldsymbol{\mathbf{x},\mathbf{y}}\in \mathbb{Z}_q^n$. For any integer $M$, we write \( [M]^+ \triangleq max\{M,0\}\) and we let \( [M] \triangleq \{1,2,\ldots,M \}\). Also $\{a\}^k \triangleq (\underbrace{a, a, \ldots, a}_{k\text{ times}})$, where the element $a$ is repeated $k$ times. For two integers $a$ and $b$, $[a : b]$ represents the set $\{a, a+1, \ldots b\}$.
 
\section{Preliminaries}
\label{preliminaries}
In this section, we review some basic concepts related to irregular-distance codes and Lee codes. Throughout the paper, we focus on q-ary Lee codes defined over the ring of integers modulo $q$, specifically considering codes over the alphabet $\mathbb{Z}_{q} = \{0,1,2,\ldots,q-1\}$.
\subsection{Irregular-Distance Codes}
Irregular-distance codes are a class of error-correcting codes with non-uniform distance constraints, designed to correct a specified set of error magnitudes, rather than all errors up to a certain weight. 
\begin{defn}[Irregular-Distance Codes\cite{Lenz}]  
	 For a given matrix $\boldsymbol{D}\in\mathbb{N}_{0}^{M\times M}$, $\mathcal{P}=\{\boldsymbol{p}_{1},\boldsymbol{p}_{2},\ldots,\boldsymbol{p}_{M}\}\subseteq\mathbb{Z}_{q}^{r}$ is called a $\mathbf{D}$-irregular-distance code ($\boldsymbol{D}$-code for short) if there exists an ordering of the codewords of $\mathcal{P}$ such that $d_H(\boldsymbol{p}_{i},\boldsymbol{p}_{j})\geq[\boldsymbol{D}]_{ij}$ for all $i,j\in[M]$.
	In addition, $N_{H}(\boldsymbol{D})$ is defined to be the smallest integer $r$ such that there exists a $\boldsymbol{D}$-code of length $r$.
	If $[\boldsymbol{D}]_{ij}=D$ for all $i\neq j$,  we write $N_{H}(\boldsymbol{D})$ as  $N_{H}(M,D)$.
\end{defn}

\subsection{Lee Codes}
We review some definitions and basic concepts related to the Lee metric here. The definition of the Lee weight of a vector in $\mathbb{Z}_q^n$ and the Lee distance between two vectors in $\mathbb{Z}_q^n$ are given below. 
\begin{defn}[Lee Weight\cite{CW}]
	 Let $\boldsymbol{x} = (x_{1}, x_{2},\ldots,x_{n})$ be a vector of length $n$ over ${\mathbb{Z}_{q}}$, where ${x_{i} \in  {Z}_{q} }$. The \textit{Lee weight} of a symbol $x_{i}$, i.e., $w_{L}(x_{i}) = min(x_{i}, q - x_{i})$. The \textit{Lee weight} of a vector $\boldsymbol{x} = (x_{1}, x_{2},\ldots, x_{n})$ is the sum of the  \textit{Lee weights} of its components, i.e. $w_{L}(\boldsymbol{x}) = \sum_{i=1}^n w_{L}(x_{i})$.
\end{defn}

\begin{defn}[Lee Distance\cite{CW}]
	 Let $\boldsymbol{x} = (x_{1}, x_{2},\ldots, x_{n})$ and  $\boldsymbol{y} = (y_{1}, y_{2},\ldots, y_{n})$ be two vectors of length $n$ over ${\mathbb{Z}_{q}}$, where $x_{i} \in  Z_{q} = \{0,1,2,\ldots,q-1\}$. The \textit{Lee distance} between any two symbols  $x_{i}$ and $x_{j}$, i.e., $d_{L}(x_{i},x_{j}) = min(|x_{i} - x_{j}|, q - |x_{i} - x_{j}|)$ and the \textit{Lee distance} between any two vectors   $\boldsymbol{x}$ and  $\boldsymbol{y}$, i.e., $d_{L}(\boldsymbol{x},\boldsymbol{y}) = \sum_{i=1}^n min(|x_{i} - y_{i}|, q - |x_{i} - y_{i}|)$. It can also be written as, $d_{L}(\boldsymbol{x},\boldsymbol{y}) = \sum_{i=1}^{n} w_{L} ((x_i - y_i) \bmod q)$.
\end{defn}

	Let $\boldsymbol{x} = (x_{1}, x_{2},\ldots, x_{n})$ be a vector of length $n$ over ${\mathbb{Z}_{q}}$. A vector $\boldsymbol{y} = (y_{1}, y_{2},\ldots, y_{n})$ is the result of at most \textit{t} errors from $\boldsymbol{x}$ if ${d_{L}(\boldsymbol{x}, \boldsymbol{y})} \le$ \textit{t}.
To compare the error-resilience of codes under different metrics, it is useful to establish a relationship between the Lee and Hamming distances. The following lemma provides bounds on the Lee distance in terms of the Hamming distance for vectors over $\mathbb{Z}_q^n$. 
\begin{lem}[\!\!\cite{CW}]
	 Let \( \boldsymbol{x}, \boldsymbol{y} \in \mathbb{Z}_q^n \) be two vectors over the ring of integers modulo \( q \), where \( q \geq 2 \). Denote by \( d_H(\boldsymbol{x}, \boldsymbol{y}) \) the Hamming distance between \( \boldsymbol{x} \) and \( \boldsymbol{y} \), and by \( d_L(\boldsymbol{x}, \boldsymbol{y}) \) the Lee distance. Then the following bounds hold:
	\[
	d_H(\boldsymbol{x}, \boldsymbol{y}) \leq d_L(\boldsymbol{x}, \boldsymbol{y}) \leq \left\lfloor \frac{q}{2} \right\rfloor d_H(\boldsymbol{x}, \boldsymbol{y}). 
	\] 
\end{lem}
A Lee code $\mathcal{C} \subseteq \mathbb{Z}_q^n$ with  $|\mathcal{C}| = M$ and minimum distance $d_L$  will be specified by $(n, M, d_L)_q$. The following result in the lemma characterizes the error-correcting capability of Lee codes in terms of their minimum Lee distance.
\begin{lem}[\!\!\cite{CW}]
	A Lee code with minimum Lee distance \( d_L \) can correct up to \( t \) errors if and only if
	\(
	d_L \geq 2t + 1.
	\)
\end{lem}

Plotkin Low-rate Average Distance Bound \cite{GW} is presented next. This upper bound on minimum distance \( d_L \) is based on the fact that the minimum distance between any pair of codewords in a code cannot exceed the average distance between all pairs of distinct codewords.
\begin{lem}[\!\cite{GW}]
 
The minimum Lee distance of an $(n, M, d_L)_q$ Lee code is bounded from above as \[ d_{L}  \le \begin{cases}
		\begin{aligned}
		&\frac{n(q^{2}-1)M}{4q(M-1)},  \text{ if } q \text{ is odd,} \\
		&\frac{nqM}{4(M-1)}, \text{ if } q \text{ is even.}	
		\end{aligned}
	\end{cases} \] 
\label{PLAD}	
\end{lem}
The following lemmas are used in the derivation of Plotkin-like bound for irregular Lee-distance codes in Section \ref{General  Results on the Optimal Redundancy}.
\begin{lem}[\!\cite{GW}]
	\label{lem_W1}
	 
	Let $x_i \in \mathbb{Z}_q = \{ 0,1,\ldots, q-1\}$ and let the sum of Lee distances from a symbol $x_i$ to all other symbols in $\mathbb{Z}_q$ be $S$, then for each $x_i \in \mathbb{Z}_q$ \[ S = \sum_{x_j = 0}^{q-1} d_L(x_i, x_j) = \begin{cases}
		\begin{aligned}
			&\frac{q^2}{4},  \text{ if } q \text{ is even}, \\
			&\frac{(q^2-1)}{4}, \text{ if } q \text{ is odd}.
		\end{aligned}
	\end{cases} \]
\end{lem}
\begin{lem}[\!\cite{GW}]
	\label{lem_W2}
	 
	Let  $\left\{ \mathbf{p}_i \right\}_{i=1}^{M} \subseteq  \mathbb{Z}_q^r $ be a Lee code of length $r$ and cardinality $M$. Then, $ \sum_{i,j: i<j} d_L(\mathbf{p}_i, \mathbf{p}_j) \le \frac{SM^2r}{2q}, $ where $ S = \sum_{x_j = 0}^{q-1} d_L(x_i, x_j)$.
\end{lem}
\section{General  Results on the Optimal Redundancy}
\label{General  Results on the Optimal Redundancy}
In this section, we first review the known results on the optimal redundancy of FCLCs. The relationship between FCLCs and irregular-Lee distance codes were recently introduced in \cite{GA}. Many of the concepts and definitions in \cite{GA} follow a structure similar to that of Lenz et al. \cite{Lenz} in their foundational work on FCCs, adapted in \cite{GA} to the Lee metric setting with appropriate modifications. We then propose a Plotkin-like bound for irregular-Lee distance codes. 
\subsection{A Connection between FCLCs and Irregular-Lee-Distance Codes }

Let \( \mathbf{u} \in \mathbb{Z}_q^k \) be the message and let \( f : \mathbb{Z}_q^k \to Im(f) = \{ f(\mathbf{u}) \mid \mathbf{u} \in \mathbb{Z}_q^k \} \) be a function computed on \( \mathbf{u} \) with expressiveness \( E = |Im(f)| \leq q^k \). The message is encoded via the encoding function
\[
\mathrm{Enc} : \mathbb{Z}_q^k \to \mathbb{Z}_q^{k+r}, \quad \mathrm{Enc}(\mathbf{u}) = (\mathbf{u}, p(\mathbf{u})),
\]
where \( p(\mathbf{u}) \in \mathbb{Z}_q^r \) is the redundancy vector and \( r \) is the redundancy. The resulting codeword \( \mathrm{Enc}(\mathbf{u}) \) is transmitted over a Lee channel, resulting in \( y \in \mathbb{Z}_q^{k+r} \) with \( d_L(\mathrm{Enc}(\mathbf{u}), y) \leq t \). The formal definition of FCLC is stated below.
\begin{defn}[Function-Correcting Lee Codes]
 An encoding function \( \mathrm{Enc} : \mathbb{Z}_q^k \to \mathbb{Z}_q^{k+r} \),  
\( \mathrm{Enc}(\mathbf{u}) = (\mathbf{u}, p(\mathbf{u})) \) defines a systematic function-correcting Lee code (FCLC for short) for the function \( f : \mathbb{Z}_q^k \to Im(f) \)  if for all \( \mathbf{u}_1, \mathbf{u}_2 \in \mathbb{Z}_q^k \) with \( f(\mathbf{u}_1) \neq f(\mathbf{u}_2) \), we have  
	\[
	d_L(\mathrm{Enc}(\mathbf{u}_1), \mathrm{Enc}(\mathbf{u}_2)) \geq 2t + 1. 
	\]
\end{defn}
Given an information vector $\mathbf{u} \in \mathbb{Z}_q^k$, the encoder computes the parity vector $\mathbf{p}(\mathbf{u}) \in \mathbb{Z}_q^r$ according to the prescribed construction and outputs the systematic codeword $ \mathrm{Enc}(\mathbf{u}) = (\mathbf{u}, p(\mathbf{u})) \in \mathbb{Z}_q^{k+r}$. The parity vector is designed such that for any pair $(\mathbf{u_1}, \mathbf{u_2})$ with $f(\mathbf{u_1}) \neq f(\mathbf{u_2})$, the corresponding codewords satisfy $d_L(\mathrm{Enc}(\mathbf{u}_1), \mathrm{Enc}(\mathbf{u}_2)) \geq 2t + 1$. 

Let $\mathbf{y}$ be the received vector over a $q$-ary symmetric channel under the Lee metric. The decoder performs minimum Lee distance decoding, i.e., $\hat{\mathbf{u}} = \displaystyle \arg\min_{\mathbf{u} \in \mathbb{Z}_q^k}
d_L(\mathbf{y}, \mathrm{Enc}(\mathbf{u}))$. The decoder then outputs the estimated function value $f(\hat{\mathbf{u}})$. Since the code satisfies $d_L(\mathrm{Enc}(\mathbf{u}_1), \mathrm{Enc}(\mathbf{u}_2)) \geq 2t + 1$ whenever  $f(\mathbf{u_1}) \neq f(\mathbf{u_2})$, it follows that for any error vector of Lee weight at most $t$, the minimum Lee distance decoder uniquely determines the correct function value.  
\begin{rem}
	By this definition, given any received vector \( \mathbf{y} \) obtained by at most \( t \) errors from \( Enc(\mathbf{u}) \), the receiver can uniquely recover \( f(\mathbf{u}) \) provided it has knowledge of the function \( f \) and the encoding function \(Enc\).
\end{rem}
The optimal redundancy of an FCLC for the function \( f \), which is the main parameter of interest in this paper, is defined next.

\begin{defn}[Optimal Redundancy]
 The optimal redundancy \( r_{L}^f(q, k, t) \) is defined as the smallest integer \( r \) for which there exists an FCLC with an encoding function \( \mathrm{Enc} : \mathbb{Z}_q^k \to \mathbb{Z}_q^{k+r} \) that enables recovery of \( f(\mathbf{u}) \) under t Lee errors.
\end{defn}
In order to determine the optimal redundancy of FCLCs, a connection between FCLCs and irregular-Lee-distance codes was established in \cite{GA}. Towards this, the definition of Lee distance requirement  matrix associated with a function \(f\) follows.

\begin{defn}[Distance Requirement Matrix] 
	\label{Def_DRM}
 Let $\mathbf{u}_1, \dots, \mathbf{u}_M \in \mathbb{Z}_{q}^{k}$. The distance requirement matrix (DRM for short)	$\mathbf{D}_f(t, \mathbf{u}_1, \dots, \mathbf{u}_M)$ of a function $f$ for $t$- error correction is defined as the $M \times M$ matrix with entries
	\[
	[\!\mathbf{D}_f\!(t,\! \mathbf{u}_1, \dots, \mathbf{u}_M)\!]_{ij}=[\mathbf{D}]_{ij}=\!
	\begin{cases}
		\begin{aligned}
			&\![2t\! +\!1\!-\!d_{L}\!(\!\mathbf{u}_i,\! \mathbf{u}_j\!)\!]^+\!, \text{if }\! f(\mathbf{u}_i)\! \neq\! f(\mathbf{u}_j), \\
			&0, \quad \text{otherwise}.
		\end{aligned}
	\end{cases}
	\]
\end{defn}
Irregular-Lee distance codes are defined formally as follows.
\begin{defn}[$\mathbf{D}_L$-code]
	\label{D-code}
	 Let $\mathbf{D} \in \mathbb{N}_0^{M \times M}$ and $\mathcal{P} = \{ \mathbf{p}_1, \mathbf{p}_2, \ldots, \mathbf{p}_M \} \subseteq \mathbb{Z}_q^r$ be a code of length $r$ and cardinality $M$. Then, $\mathcal{P} = \{ \mathbf{p}_1, \mathbf{p}_2, \ldots, \mathbf{p}_M \}$ is a $\mathbf{D}$-irregular-Lee-distance code ($\mathbf{D}_L$-code for short), if there exists an ordering of the codewords of $\mathcal{P}$ such that $d_{L}(\mathbf{p}_i, \mathbf{p}_j) \geq [\mathbf{D}]_{ij}$ for all $i, j \in [M]$.
\end{defn}
The smallest integer $r$ such that there exists a $\mathbf{D}_{L}$-code of length $r$ is denoted by $N_{L}(\mathbf{D})$. If $[\mathbf{D}]_{ij} = D$ for all $i \neq j$, we simply write $N_{L}(M,D)$.
By definition, $\mathbf{D}_{L}$-code imposes individual Lee distance constraints between every pair of codewords.
For two distinct information vectors $\mathbf{u}_i$ and $\mathbf{u}_j$, the corresponding parity vectors $\mathbf{p}_i$ and $\mathbf{p}_j$ must satisfy $d_{L}(\mathbf{p}_i, \mathbf{p}_j) \geq [\mathbf{D}]_{ij}$, where the entry $[\mathbf{D}]_{ij}$ is determined by the Lee distance $d_{L}(\mathbf{u}_i, \mathbf{u}_j)$, and whether $f(\mathbf{u}_i) = f(\mathbf{u}_j)$ or $f(\mathbf{u}_i) \neq f(\mathbf{u}_j)$. The distance requirement between two distinct parity vectors is zero only if the corresponding information vectors evaluate to the same function value, or information vectors with different function values but separated by a Lee distance of at least $2t+1$. In all other cases, the DRM assigns a strictly positive distance requirement. Thus, DRM is not arbitrarily chosen. It is systematically derived from the function-separation requirement $d_{L}((\mathbf{u}_i, \mathbf{p}_i), (\mathbf{u}_j, \mathbf{p}_j)) \ge 2t+1$ whenever $f(\mathbf{u}_i) \neq f(\mathbf{u}_j)$. Therefore, the DRM specifies the Lee distance constraint between each pair of parity vectors to guarantee correction of up to $t$ errors in the function value. This differs fundamentally from classical(regular) minimum-distance codes. In a regular distance systematic code, every pair of distinct codewords must satisfy a uniform minimum distance constraint $d_{min}$. In contrast, for irregular distance codes arising in the function-correcting setting, each distinct pair of parity vectors must satisfy its own specific distance requirement as prescribed by the DRM. Accordingly, the systematic  code design problem differs in nature: For regular distance codes, the objective is to select $M$ parity vectors such that every codeword thus formed is separated by at least $d_{min}$. For function-correcting codes, the objective is to select $M$ parity vectors that satisfy the individual pairwise distance requirements specified by the DRM. Since only certain parts require large separation, the irregular structure often allows strictly smaller redundancy than a regular Lee code with the same minimum distance. An example for a distance requirement matrix is given below.
\begin{example}
	\label{Ex_DRM}
 Let $\mathbf{u}_i = (\mathbf{u}_{i1}, \mathbf{u}_{i2}) \in \mathbb{Z}_5^2, \;\forall i \in [5]$ and $f(\mathbf{u}_i) = \mathrm{w_L}(\mathbf{u}_{i2})$. The distance requirement matrix $\mathbf{D}_f(t=1, \mathbf{u}_1, \ldots, \mathbf{u}_5)$ for $f(\mathbf{u}_1)= 0, f(\mathbf{u}_2)= 1, f(\mathbf{u}_3)= 2, f(\mathbf{u}_4)= 2, f(\mathbf{u}_5)= 1$ with $ \mathbf{u}_1 = 00, \mathbf{u}_2 = 01, \mathbf{u}_3 = 02, \mathbf{u}_4 = 03, \mathbf{u}_5 = 04$, respectively and $t=1$  is a $5 \times 5$ matrix  given by, 
	
			$$\mathbf{D}_f(t=1, \mathbf{u}_1, \ldots, \mathbf{u}_5)=\begin{blockarray}{c@{\hskip 4pt}*{5}{c}}
				& \mathbf{u}_1 & \mathbf{u}_2 & \mathbf{u}_3 & \mathbf{u}_4 & \mathbf{u}_5 \\
				\begin{block}{c@{\hskip 4pt}[*{5}{c}]}
					\mathbf{u}_1 \text{ } & 0 & 2 & 1 & 1 & 2 \\
					\mathbf{u}_2 \text{ } & 2 & 0 & 2 & 1 & 0 \\
					\mathbf{u}_3 \text{ } & 1 & 2 & 0 & 0 & 1 \\
					\mathbf{u}_4 \text{ } & 1 & 1 & 0 & 0 & 2 \\
					\mathbf{u}_5 \text{ } & 2 & 0 & 1 & 2 & 0 \\
				\end{block}
			\end{blockarray}$$
	
\end{example}
This matrix is clearly non-uniform, some parity pairs require Lee distance $2$, some require Lee distance $1$, and others require no separation at all. This non-uniform structure clearly distinguishes irregular Lee distance codes from regular Lee codes, where every distinct pair must satisfy the same minimum distance requirement. For comparison, a regular Lee code with minimum distance $d_{min}=3$ would require $[\mathbf{D}]_{ij} \ge 3$, if $i \neq j$ and $[\mathbf{D}]_{ij}=0$, if $i = j$. Such a distance requirement matrix $\mathbf{D}(t=1, \mathbf{u}_1, \ldots, \mathbf{u}_5)$ is shown next, where every distinct pair of codewords $(\mathbf{c}_{i}, \mathbf{c}_{j})$, $\forall \mathbf{c}_{i}, \mathbf{c}_{j} \in \mathcal{C}$ must satisfy the same distance constraint.
$$\mathbf{D}(t=1, \mathbf{u}_1, \ldots, \mathbf{u}_5)=\begin{blockarray}{c@{\hskip 4pt}*{5}{c}}
			& \mathbf{c}_1 & \mathbf{c}_2 & \mathbf{c}_3 & \mathbf{c}_4 & \mathbf{c}_5 \\
			\begin{block}{c@{\hskip 4pt}[*{5}{c}]}
				\mathbf{c}_1 \text{ } & 0 & 3 & 3 & 3 & 3 \\
				\mathbf{c}_2 \text{ } & 3 & 0 & 3 & 3 & 3 \\
				\mathbf{c}_3 \text{ } & 3 & 3 & 0 & 3 & 3 \\
				\mathbf{c}_4 \text{ } & 3 & 3 & 3 & 0 & 3 \\
				\mathbf{c}_5 \text{ } & 3 & 3 & 3 & 3 & 0 \\
			\end{block}
		\end{blockarray}$$
 A $\mathbf{D}_{L}$-code for Example~\ref{Ex_DRM} is given next.
\begin{example}
	\label{Ex_D_code}
 Consider $\mathcal{P} = \{ p_1 = 0, \; p_2 = 2, \; p_3 = 4, \; p_4 = 1, \; p_5 =3 \}$.  It can be easily verified that, for this $\mathcal{P}$ taken in the same order, the condition  $d_{L}(p_i, p_j) \geq [\mathbf{D}]_{ij}$, $\forall i, j \in [5]$ is satisfied for the distance requirement matrix in Example~\ref{Ex_DRM}. Therefore, $\mathcal{P}$ is a $\mathbf{D}_{L}$-code of length $r = 1$.
\end{example}

Given these definitions, we proceed to establish a link between FCLCs and irregular-Lee-distance codes. 

\begin{thm}[\!\!\cite{GA}]
 For any function $f : \mathbb{Z}_q^k \rightarrow \text{Im}(f)$, 
	$
	r_{L}^f(q, k, t) = N_{L}\left(\mathbf{D}_f(t, \mathbf{u}_1, \dots, \mathbf{u}_{q^k})\right),
	$
	where $\{\mathbf{u}_1, \dots, \mathbf{u}_{q^k}\} = \mathbb{Z}_q^k$ denotes the set of all $q$-ary vectors of length $k$.
\end{thm}

\subsection{Simplified Bounds on Optimal Redundancy}
In this subsection, we review existing bounds on optimal redundancy \( r_{L}^f(q, k, t) \) of FCLCs given in \cite{GA} and also propose a Plotkin-like bound for irregular-Lee-distance codes.
\begin{cor}[\!\!\cite{GA}]
 Let $\mathbf{u}_1, \ldots, \mathbf{u}_M \in \mathbb{Z}_q^k$ be arbitrary different vectors. Then, the redundancy of an FCLC is at least
	\[
	r_{L}^f(q, k, t) \geq N_{L}\left(\mathbf{D}_f(t, \mathbf{u}_1, \ldots, \mathbf{u}_M)\right).
	\]
	
	For any function $f$ with $|\text{Im}(f)| \geq 2$,
$
	r_{L}^f(q, k, t) \geq  N_{L}(2, 2t) = \left\lceil \frac{2t}{\lfloor \frac{q}{2} \rfloor} \right\rceil.
	$
\end{cor}
A simplified upper bound on \( r_f(q, k, t) \) by considering a representative subset of information vectors corresponding to distinct function values was given in \cite{GA}. Toward this, the definition of the Lee distance between two function values follows.
\begin{defn}[Function Distance]
 The Lee distance between two function values \( f_1, f_2 \in \mathrm{Im}(f) \) is defined as the minimum Lee distance between any pair of information vectors that evaluate to \( f_1 \) and \( f_2 \), i.e.,
	\[
	d_{L}^f(f_1, f_2) \triangleq \! \min_{\mathbf{u}_1, \mathbf{u}_2 \in \mathbb{Z}_q^k} \! d_{L}(\mathbf{u}_1, \mathbf{u}_2) \; \text{s.t.} \; f(\mathbf{u}_1) = f_1, f(\mathbf{u}_2) = f_2.
	\]
\end{defn}
Note that \( d_{L}^f(f_1, f_1) = 0 \), \( \forall f_1 \in \mathrm{Im}(f) \). Based on this, the definition of the Lee function-distance matrix of \( f \) follows.
\begin{defn}[Function-Distance Matrix ] 
	\label{Def_FDM}
 The function-distance matrix of a function \( f \) is an \( E \times E \) matrix denoted by \( \mathbf{D}_f(t, f_1, \ldots, f_E)\) with entries \( \left[ \mathbf{D}_f(t, f_1, \ldots, f_E) \right]_{ij} = \left[2t + 1 - d_{L}^f(f_i, f_j) \right]^+, \; \text{if } i \neq j\)
	and \(\left[ \mathbf{D}_f(t, f_1, \ldots, f_E) \right]_{ii} = 0\), for $t$ error correction.
\end{defn}
An example for a function distance matrix is given next.
\begin{example}
	\label{Ex_FDM}
Let $\mathbf{u}_i = (\boldsymbol{u}_{i1}, \boldsymbol{u}_{i2}) \in \mathbb{Z}_5^2$ and $f(\mathbf{u}_i) = \mathbf{u}_{i1}$,  $ \forall i \in [M]$, where $ M = 25$. For this function,   $Im(f) = \{0, 1, 2, 3, 4\}$. The corresponding function distance matrix for $t=1$ is given by, 
	\[
	\mathbf{D}_f(t=1, 0, 1, 2, 3, 4) =
	\begin{bmatrix}
		0 & 2 & 1 & 1 & 2 \\
		2 & 0 & 2 & 1 & 1 \\
		1 & 2 & 0 & 2 & 1 \\
		1 & 1 & 2 & 0 & 2 \\
		2 & 1 & 1 & 2 & 0 \\
	\end{bmatrix}
	.\]
\end{example}

\begin{thm}[\!\!\cite{GA}]
	\label{Thm_UB}
 For any arbitrary function $f : \mathbb{Z}_q^k \rightarrow \mathrm{Im}(f)$,
	$r_{L}^f(q, k, t) \leq N_{L}\left(\mathbf{D}_f(t, f_1, \ldots, f_E)\right). \quad$
\end{thm}

In certain cases, the bound provided in Theorem \ref{Thm_UB} is tight. One such significant case is described in the following corollary.
\begin{cor}[\!\!\cite{GA}]
	\label{Cor_OP}
  If there exists a set of representative information vectors $\mathbf{u}_1, \ldots, \mathbf{u}_E$ such that \( \{ f(\mathbf{u}_1), \ldots, f(\mathbf{u}_E) \} = \mathrm{Im}(f)
	\; \text{and} \; \mathbf{D}_f(t, \mathbf{u}_1, \ldots, \mathbf{u}_E) = \mathbf{D}_f(t, f_1, \ldots, f_E), \) then $ r_L^f(q, k,t) = \\ N_{L}(\mathbf{D}_f(t, f \ldots, f_E)). $
\end{cor}
Although the bound in Theorem \ref{Thm_UB} is not always tight, it is often more practical to work with the function distance matrix $\mathbf{D}_f(t, f_1, \ldots, f_E)$ rather than the full distance requirement matrix $D_f(t, \boldsymbol{u}_1, \ldots, \boldsymbol{u}_{q^k})$, especially when the number of distinct function values $E$ is small. Computing $N_{L}(\mathbf{D}_f(t, f_1, \ldots, f_E))$ is typically much more tractable in such cases.
\subsection{A variant of Plotkin Bound on \( N_{L}(\mathbf{D}) \)} 
 The traditional Plotkin bound for error-correcting codes in the Lee metric is derived under the requirement that the Lee distance between any two codewords must be at least the minimum distance of the code. In contrast, for function-correcting Lee codes (FCLCs), the requirement is relatively more relaxed: the Lee distance between two codewords needs to be at least the minimum distance only when the corresponding messages produce different function values. Consequently, bounds derived under the stricter constraints of classical error-correcting codes, such as the traditional Plotkin bound, are not directly applicable to FCLCs. This highlights the need for bounds specifically tailored to the FCLC setting. To clarify this point, we illustrate the difference through the following example.
 \begin{example}
 	Let \( \mathbf{u}=(u_1 u_2) \in \mathbb{Z}_5^2 \) be the message. Consider any binary function \( f : \mathbb{Z}_5^2 \to \{0,1\} \) and consider the following encoding.  $ \mathrm{Enc}(\boldsymbol{u}) = (\boldsymbol{u}, \{p_{\boldsymbol{u}}\}^{3} )$, where
 	\begin{equation*}
 		p_{\boldsymbol{u}} = \begin{cases}
 			\begin{aligned}
 				&2,  \text{              } f(\boldsymbol{u}) = 0, \\
 				&0,  \text{              } f(\boldsymbol{u}) = 1,
 			\end{aligned}
 		\end{cases}
 	\end{equation*}	
 	and \(\{p_{\boldsymbol{u}}\}^{3}\) stands for the  \( 3 \)-fold repetition of the parity symbol \(p_{\boldsymbol{u}}\). This encoding  results in an FCLC for the considered binary function for $t=3$. To prove this, we have to show that for all \( \mathbf{u}, \mathbf{u'} \in \mathbb{Z}_5^2 \) with \( f(\mathbf{u}) \neq f(\mathbf{u'}) \), the above encoding ensures \(d_L(\mathrm{Enc}(\mathbf{u}), \mathrm{Enc}(\mathbf{u'})) \geq 7 \). Whenever \( f(\mathbf{u}) \neq f(\mathbf{u'}) \), WLOG assume $f(\mathbf{u})=0$ and $f(\mathbf{u'})=1$. In this case,  we have $d_{L}(\mathrm{Enc}(\boldsymbol{u}), \mathrm{Enc}(\boldsymbol{u'})) = d_{L}(\boldsymbol{u}, \boldsymbol{u'}) + d_{L}(\{{p}_{\boldsymbol{u}}\}^{3}, \{{p}_{\boldsymbol{u'}}\}^{3}) \ge 1 + 3 \times d_L(2,0) = 1 +3 \times 2 =7$. Therefore, the above encoding yields an FCLC for any binary function \( f : \mathbb{Z}_5^2 \to \{0,1\} \) and $t=3$, achieving a redundancy $r=3$. 
 	
 	Now consider the traditional Plotkin bound for Lee metric given in Lemma $3$. For odd $q$, we have $d_L \le \frac{(k+r)(q^{2}-1)M}{4q(M-1)}$. Equivalently, $r \ge \frac{4q(M-1)d_L}{q^{2}-1)M} -k$. For $q=5$, $k=2$ and $t=3$, we have $M=25$ and $d_L=7$. Substituting these values we obtain $r \ge 3.6$. That is, by the traditional Plotkin bound, we need at least a redundancy of $r=4$ to ensure $3$ Lee error correction. However, for the binary function, we obtained FCLC ensuring $3$ Lee error correction with a lesser redundancy of $3$.     
 \end{example} Therefore, we extend the classical Plotkin bound to the setting of irregular Lee-distance codes, as stated in the following theorem.
\begin{thm}
	\label{Thm_PB}
	For any distance matrix $\mathbf{D} \in \mathbb{N}_{0}^{M \times M}$, \begin{equation}
		N_L(\mathbf{D}) \ge \begin{cases}
			\begin{aligned}
				&\frac{8}{M^2q}\sum_{i,j: i<j} [\mathbf{D}]_{ij},  \text{ if } q \text{ is even,} \\
				&\frac{8q}{M^2(q^2-1)}\sum_{i,j: i<j} [\mathbf{D}]_{ij},  \text{ if } q \text{ is odd}.
			\end{aligned}
		\end{cases}
		\label{Eq_PB1}
	\end{equation} 
\end{thm}

\begin{IEEEproof}
	Let $N_L(\mathbf{D}) = r$. Let $\left\{ \mathbf{p}_i \right\}_{i=1}^{M}$ be codewords of a $\mathbf{D}_L$-code of length $r$. Since $\left\{ \mathbf{p}_i \right\}_{i=1}^{M}$ form a $\mathbf{D}_L$-code, by definition we have $[\mathbf{D}]_{ij} \le d_L(\mathbf{p}_i, \mathbf{p}_j)  \quad \forall i, j$. Therefore, \begin{equation}
		\sum_{i,j: i<j} [\mathbf{D}]_{ij} \le \sum_{i,j: i<j} d_L(\mathbf{p}_i, \mathbf{p}_j). 
		\label{Eq_PB2}
	\end{equation}
	Combining Lemma~\ref{lem_W1} and Lemma~\ref{lem_W2}, we obtain \begin{equation}
		\sum_{i,j: i<j} d_L(\mathbf{p}_i, \mathbf{p}_j) \le \frac{SM^2r}{2q} = \begin{cases}
			\begin{aligned}
				&\frac{qM^2r}{8},  \text{ if } q \text{ is even,} \\
				&\frac{(q^2-1)M^2r}{8q},  \text{ if } q \text{ is odd}.
			\end{aligned}
		\end{cases}
		\label{Eq_PB3}
	\end{equation} 
	From (\ref{Eq_PB2}) and (\ref{Eq_PB3}),  \[ \sum_{i,j: i<j} [\mathbf{D}]_{ij} \le \begin{cases}
		\begin{aligned}
			&\frac{qM^2r}{8},  \text{ if } q \text{ is even,} \\
			&\frac{(q^2-1)M^2r}{8q},  \text{ if } q \text{ is odd}.
		\end{aligned}
	\end{cases}  \]
	Rearranging we obtain, \[ N_L(\mathbf{D}) \ge \begin{cases}
		\begin{aligned}
			&\frac{8}{M^2q}\sum_{i,j: i<j} [\mathbf{D}]_{ij},  \text{ if } q \text{ is even,} \\
			&\frac{8q}{M^2(q^2-1)}\sum_{i,j: i<j} [\mathbf{D}]_{ij},  \text{ if } q \text{ is odd}.
		\end{aligned}
	\end{cases}\] 
\end{IEEEproof}
\begin{rem}
For $q=2$ in (\ref{Eq_PB1}), we obtain a bound for irregular-distance codes in the Hamming metric, i.e., $ N(\mathbf{D}) \ge \frac{4}{M^2}\sum_{i,j: i<j} [\mathbf{D}]_{ij}$. Comparing this with the Plotkin-like bound proposed in \cite{Lenz} for irregular-distance codes in the Hamming metric, we observe that both bounds are the same for even values of $M$. For odd values of $M$, the bound in \cite{Lenz} is tighter than ours by a factor of $\frac{M^2}{M^2-1}$. For large $M$, both bounds become asymptotically equal. 
\end{rem}

For regular-distance codes with minimum Lee distance \( d_L \), the total pairwise distance satisfies  
\(
 \frac{M(M-1)}{2} d_L \le \sum_{i < j} [\mathbf{D}]_{ij}.
\)
This  yields the Plotkin-like bound stated in Lemma~\ref{PLAD}. 

%\subsection{A variant of Gilbert-Varshamov's Bound \( N_{L}(\mathbf{D}) \)}  
%The generalization of the Gilbert-Varshamov bound on codes with irregular-Lee distance %requirements is discussed next. 
%\begin{lem} [\cite{GA}]
% For any distance matrix $ \mathbf{D} \in \mathbb{N}_0^{M \times M}$, and any %permutation $\pi : [M] \rightarrow [M]$ \vspace{-0.2cm} 
%	\[
%	N_L(\mathbf{D}) \leq \min_{r \in \mathbb{N}} \left\{ r : q^r > \max_{j \in [M]} \sum_{i=1}^{j-1} V\left(r, [\mathbf{D}]_{\pi(i)\pi(j)} - 1\right) \right\},
%	\] where \( V(r,t) \) is the volume of a Lee sphere with radius $t$ over vectors of length %$r$.
%\end{lem}
\section{Explicit Code Constructions for Lee Weight, Lee weight distribution, and Modular sum Functions}
\label{Lee Functions}
In this section, we study three important classes of functions, namely the Lee weight function, the Lee weight distribution function, and the modular sum function. 
\subsection{Lee weight function}
The Lee weight of a vector provides a measure of its total deviation from the zero vector under the Lee metric, providing a natural measure of error magnitude.
\begin{defn}[Lee weight function]
 A Lee weight function is defined as  \( f(\boldsymbol{u}) = \mathrm{w_L}(\boldsymbol{u}) \), where \( \boldsymbol{u} \in \mathbb{Z}_q^k \text{ and } k \in \mathbb{N} \). 
\end{defn}
The expressiveness of the Lee weight function is given by,
\( E = \left| \mathrm{Im}(\mathrm{w_L}) \right| = k \left\lfloor \frac{q}{2} \right\rfloor + 1 \).
 For Lee weight function, a set of representative information vectors can be identified for which the distance requirement matrix coincides with the function distance matrix, as shown in Lemma \ref{Lem_Cor2}. For simplicity, throughout this section, we denote the function distance matrix \( \mathbf{D}_{\mathrm{w_L}}(t, f_1, \ldots, f_E) \) as \( \mathbf{D}_{\mathrm{w_L}}(E,t) \). 

\begin{lem}[\!\!\cite{GA}]
	\label{Lem_Cor2}
 Let \( f(\mathbf{u}) = \mathrm{w}_L(\mathbf{u}) \) be the Lee weight function on \( \mathbf{u} \in \mathbb{Z}_q^k \). Consider the set of \( E = k \left\lfloor \frac{q}{2} \right\rfloor + 1 \) representative information vectors
	\(
	(\mathbf{u}_1, \mathbf{u}_2, \ldots, \mathbf{u}_E) = (0^k, 0^{k-1}1, \ldots, 0^{k-1}\left\lfloor \tfrac{q}{2} \right\rfloor, 0^{k-2}1 \\ \left\lfloor \tfrac{q}{2} \right\rfloor \ldots, 0^{k-2}\left\lfloor \tfrac{q}{2} \right\rfloor\left\lfloor \tfrac{q}{2} \right\rfloor, 0^{k-3}1{\left\lfloor \tfrac{q}{2} \right\rfloor}^2 \ldots,{(\left\lfloor \tfrac{q}{2} \right\rfloor - 1)}{\left\lfloor \tfrac{q}{2} \right\rfloor}^{k-1} ,{\left\lfloor \tfrac{q}{2} \right\rfloor}^k)
	\)
	such that \( f(\mathbf{u}_i) = i - 1 \in \mathrm{Im}(f) \) for all \( i \in [E] \). Then, for this set of vectors, the distance requirement matrix and the function distance matrix are identical. Consequently, by Corollary \ref{Cor_OP}, the optimal redundancy of FCLCs satisfies
	$
	r_L^\mathrm{w}(q, k,t) = N_{L}(\mathbf{D}_{\mathrm{w}_L}(E,t)),
	$
	where \( \mathbf{D}_{\mathrm{w}_L}(E,t) \) denotes the function distance matrix for \( t \)-error correction, whose $(i,j)$-th entry is given by
	$$[\mathbf{D}_{\mathrm{w}_L}(E,t))]_{ij}=\begin{cases}
		0 & \text{ if } i=j\\
		[2t+1-|i-j|]^+ & \text{ if } i\neq j.
	\end{cases}$$
\end{lem}
The following two examples illustrate Lemma~\ref{Lem_Cor2} by demonstrating that the distance requirement matrix coincides with the function distance matrix.
\begin{example}[Illustration of Lemma~\ref{Lem_Cor2} for \( q = 5, k = 2, t = 1 \)]
Consider the Lee weight function over \( \mathbb{Z}_5^2 \). The image of the function is \(
	\mathrm{Im}(\mathrm{w}_L) = \{0, 1, 2, 3, 4\},\) Therefore,  the no. of representative vectors required is $5$. We choose \(\mathbf{u}_1 = 00,\; \mathbf{u}_2 = 01, \; \mathbf{u}_3 = 02,\; \mathbf{u}_4 = 12,\;\mathbf{u}_5 = 22 \). The corresponding Lee weights of these vectors are \( 0,1,2,3,4 \),  respectively, i.e., \( f(\mathbf{u}_i) = i - 1 \). For the chosen representative vectors $\{\mathbf{u}_1, \mathbf{u}_2, \mathbf{u}_3, \mathbf{u}_4, \mathbf{u}_5\}$ and \( t = 1 \), the distance requirement matrix is given by,
	\[
	\mathbf{D}_f(t=1, \mathbf{u}_1, \ldots, \mathbf{u}_5) =
	\begin{bmatrix}
		0 & 2 & 1 & 0 & 0 \\
		2 & 0 & 2 & 1 & 0 \\
		1 & 2 & 0 & 2 & 1 \\
		0 & 1 & 2 & 0 & 2 \\
		0 & 0 & 1 & 2 & 0 \\
	\end{bmatrix}
	.\]
	For function values $\{f(\mathbf{u}_1) =f_1, f(\mathbf{u}_2)=f_2, f(\mathbf{u}_3)=f_3, f(\mathbf{u}_4)=f_4, f(\mathbf{u}_5)=f_5\}$ and \( t = 1 \), the function distance matrix is given by,
	\[
	\mathbf{D}_f(t=1, f_1, \ldots, f_5) =
	\begin{bmatrix}
		0 & 2 & 1 & 0 & 0 \\
		2 & 0 & 2 & 1 & 0 \\
		1 & 2 & 0 & 2 & 1 \\
		0 & 1 & 2 & 0 & 2 \\
		0 & 0 & 1 & 2 & 0 \\
	\end{bmatrix}
	.\]
	The two matrices are identical, thereby validating Lemma~\ref{Lem_Cor2}. Hence,
	\(
	r_L^\mathrm{w}(5,2,1) = N_L\left( \mathbf{D}_{\mathrm{w_L}}(\mathrm{5},1) \right).
	\)
\end{example}
\begin{example}[Illustration of Lemma~\ref{Lem_Cor2} for \( q = 5, k = 3, t = 1 \)]
	Consider the Lee weight function over \( \mathbb{Z}_5^3 \). The image of the function, \( \mathrm{Im}(\mathrm{w}_L) = \{0, 1, 2, 3, 4, 5, 6\}, \) Therefore, the no. of representative vectors required is  \(7 \). We choose
	\(
	\mathbf{u}_1 = 000,\;
	\mathbf{u}_2 = 001,\;
	\mathbf{u}_3 = 002,\;
	\mathbf{u}_4 = 012,\;
	\)
	\(
	\mathbf{u}_5 = 022,\;
	\mathbf{u}_6 = 122,\;
	\mathbf{u}_7 = 222.
	\)
	The corresponding Lee weights of these vectors are \( 0,1,2,3,4,5,6 \), i.e.,  \( f(\mathbf{u}_i) = i - 1 \).
	For the chosen representative vectors $\{\mathbf{u}_1, \mathbf{u}_2, \mathbf{u}_3, \mathbf{u}_4, \mathbf{u}_5, \mathbf{u}_6, \mathbf{u}_7\}$ and \( t = 1 \), the distance requirement matrix is given by,
	\[
	\mathbf{D}_f(t=1, \mathbf{u}_1, \ldots, \mathbf{u}_7) =
	\begin{bmatrix}
		0 & 2 & 1 & 0 & 0 & 0 & 0 \\
		2 & 0 & 2 & 1 & 0 & 0 & 0 \\
		1 & 2 & 0 & 2 & 1 & 0 & 0 \\
		0 & 1 & 2 & 0 & 2 & 1 & 0 \\
		0 & 0 & 1 & 2 & 0 & 2 & 1 \\
		0 & 0 & 0 & 1 & 2 & 0 & 2 \\
		0 & 0 & 0 & 0 & 1 & 2 & 0 \\
	\end{bmatrix}
	.\]
	For function values $\{f(\mathbf{u}_1)=f_1, \ldots, f(\mathbf{u}_7)=f_7\}$ and \( t = 1 \), the function distance matrix is given by,
	\[
	\mathbf{D}_f(t=1, f_1, \ldots, f_7) =
	\begin{bmatrix}
		0 & 2 & 1 & 0 & 0 & 0 & 0 \\
		2 & 0 & 2 & 1 & 0 & 0 & 0 \\
		1 & 2 & 0 & 2 & 1 & 0 & 0 \\
		0 & 1 & 2 & 0 & 2 & 1 & 0 \\
		0 & 0 & 1 & 2 & 0 & 2 & 1 \\
		0 & 0 & 0 & 1 & 2 & 0 & 2 \\
		0 & 0 & 0 & 0 & 1 & 2 & 0 \\
	\end{bmatrix}
	.\]
	The matrices match, validating Lemma~\ref{Lem_Cor2}  for this case as well. Therefore,
	$
	r_L^\mathrm{w}(5, 3,1) = N_L\left( \mathbf{D}_{\mathrm{w_L}}(\mathrm{7},1) \right).
$
\end{example}

We now present a construction for FCLCs designed for the Lee weight function and derive the corresponding redundancy. The construction is based on the idea of assigning the same parity symbol to all information vectors $\mathbf{u}$ that have the same $f(\boldsymbol{u}) \bmod q$.  Conversely, distinct parity symbols are assigned to sets of vectors with different $f(\boldsymbol{u}) \bmod q$.
\begin{cons}
	\label{Cons1}
 For \( q \ge 5, \boldsymbol{u} \in \mathbb{Z}_q^k \) and a function $f : \mathbb{Z}_q^k \rightarrow \mathrm{Im}(f)$. The encoding of $\boldsymbol{u}$ is given by
	$ \mathrm{Enc}(\boldsymbol{u}) = (\boldsymbol{u}, \{p_{\boldsymbol{u}}\}^{t})$ when $q$ is odd, and $ \mathrm{Enc}(\boldsymbol{u}) = (\boldsymbol{u}, \{p_{\boldsymbol{u}}\}^{t},p'_{\boldsymbol{u}} )$ when $q$ is even, where
	\begin{equation}
		\label{Eq_Cons1}
		p_{\boldsymbol{u}} =
		\left\{
		\begin{array}{ll}
			\displaystyle (2f(\boldsymbol{u})) \bmod q, & \text{if }  \text{ q is odd, }\\[6pt]
			\displaystyle (2f(\boldsymbol{u})) \bmod q, & \text{if } 0 \le f(\boldsymbol{u}) \bmod q \le  \left(  \frac{q}{2}  - 1 \right) \\& \text{and q is even, } \\[6pt]
			\displaystyle (2f(\boldsymbol{u}) + 1) \bmod q, & \text{if }  \frac{q}{2}  \le f(\boldsymbol{u}) \bmod q \le (q - 1) \\&\text{ and q is even, } 
		\end{array}
		\right.
	\end{equation}
	\begin{equation}
		\label{Eqb_Cons}
		%\scriptsize
		p'_{\boldsymbol{u}} =
		\left\{
		\begin{array}{ll}
			\displaystyle \frac{q}{2}, & \text{if }  f(\boldsymbol{u}) \bmod q = (q-1),\\[6pt]
			\displaystyle 0, & \text{otherwise, } 
		\end{array}
		\right.
	\end{equation}
	and \(\{p_{\boldsymbol{u}}\}^{t}\) stands for the  \( t \)-fold repetition of the parity symbol \(p_{\boldsymbol{u}}\). 
\end{cons}
 \textit{Construction~\ref{Cons1}} can be used to design FCLCs for the Lee weight function as established in the following lemma.
\begin{lem}
	\label{Lem_Cons1_FCLC}
 Construction~\ref{Cons1} yields an FCLC for the Lee weight function, achieving a redundancy of $t$ when $q$ is odd and redundancy $t+1$ when $q$ is even, when $q \ge 5$ and $t \le \frac{q - 3}{2}$.
\end{lem}
Proof of Lemma \ref{Lem_Cons1_FCLC} is given in Appendix \ref{appendix:Lem_Cons1_FCLC}. 
Based on Lemma~\ref{Lem_Cor2}, we derive lower bounds on the optimal redundancy $r_L^\mathrm{w}(q, k,t)$, by applying the Plotkin-like bound from Theorem~\ref{Thm_PB}, as stated next in Corollary~\ref{Cor_PB_LW}.
\begin{cor}
	\label{Cor_PB_LW}
 For any $q \ge 5$ and \( t = \left\lfloor \frac{q-3}{2} \right\rfloor \), we have \\
\begin{equation*}
	\small
	r_L^\mathrm{w}(q, k,t) \ge \begin{cases}
		\begin{aligned}
			&\frac{8q}{E^2 (q^2-1) } t(2t+1)(E-\frac{2(t+1)}{3}),  \text{ if } q \text{ is odd,} \\
			&\frac{8}{E^2 q}  t(2t+1)(E-\frac{2(t+1)}{3}),  \text{ if } q \text{ is even},
		\end{aligned}
	\end{cases}
\end{equation*} 
	where \(E = k\left\lfloor \frac{q}{2} \right\rfloor +1 \).
\end{cor}

\begin{IEEEproof}
By Lemma~\ref{Lem_Cor2}, for the chosen representative information vectors, the distance requirement matrix and the function distance matrix are identical,  as shown in Figure~\ref{Fig_FDM_LW} for \( t = \frac{q-3}{2} \). 
\begin{figure*}[!htbp]
	\centering
	\setlength{\arraycolsep}{3pt}   % reduce spacing between columns
	\small % compact font
	\[
	\begin{blockarray}{c@{\hskip 4pt} *{9}{c}}
		& 0 & 1 & \cdots & (2t-1) & 2t & (2t+1) & \cdots & k\!\lfloor \frac{q}{2}\!\rfloor -1 & k\!\lfloor \tfrac{q}{2} \rfloor \\
		\begin{block}{c@{\hskip 4pt}[*{9}{c}]}
			0   & 0 & (2t+1-1) & \cdots & 2 & 1 & 0 & \cdots & 0 & 0 \\[3pt]
			1   & (2t+1-1) & 0 & \cdots & 3 & 2 & 1 & \cdots & 0 & 0 \\[3pt]
			\vdots & \vdots & \vdots & \cdots & \ddots & \ddots & \ddots & \cdots & \vdots & \vdots \\[3pt]
			(2t-1) & 2 & 3 & \cdots & 0 & (2t+1-1) & (2t+1-2) & \cdots & 0 & 0 \\[3pt]
			2t & 1 & 2 & \cdots & (2t+1-1) & 0 & (2t+1-1) & \cdots & 1 & 0 \\[3pt]
			(2t+1) & 0 & 1 & \cdots & (2t+1-2) & (2t+1-1) & 0 & \cdots & 2 & 1 \\[3pt]
			\vdots & \vdots & \vdots & \cdots & \ddots & \ddots & \ddots & \cdots & \vdots & \vdots \\[3pt]
			k\!\lfloor \tfrac{q}{2} \rfloor -1 & 0 & 0 & \cdots & 0 & 1 & 2 & \cdots & 0 & (2t+1-1) \\[3pt]
			k\!\lfloor \tfrac{q}{2} \rfloor & 0 & 0 & \cdots & 0 & 0 & 1 & \cdots & (2t+1-1) & 0 \\
		\end{block}
	\end{blockarray}
	\]
	\caption{Function distance matrix $\mathbf{D}_{\mathrm{w}_L}(E,t)$ for $t = \frac{q-3}{2}$.}
	\label{Fig_FDM_LW}
\end{figure*}
	
From the matrix $\mathbf{D}_{\mathrm{w}_L}(E,t)$ in Figure~\ref{Fig_FDM_LW}, we observe that above the main diagonal, there are $(E-1)$ entries of $(2t+1-1)$, $(E-2)$ entries of $(2t+1-2)$, and so on, down to $(E-2t)$ entries of $1$. The sum of the entries above the main diagonal of matrix $\mathbf{D}_{\mathrm{w}_L}(E,t)$ is given by,   $\sum_{i,j: i<j} [\mathbf{D}]_{ij} = (E-1)(2t+1-1) + (E-2)(2t+1-2) + \ldots + (E-2t) (1)$. This expression simplifies to $\sum_{i,j: i<j} [\mathbf{D}]_{ij} = t(2t+1)(E-\frac{2(t+1)}{3})$. Substituting this result into (\ref{Eq_PB1}) gives the Plotkin-like bound for $q \ge 5$ and \( t = \left\lfloor \frac{q-3}{2} \right\rfloor \).
\end{IEEEproof}

The following examples illustrate Lemma \ref{Lem_Cons1_FCLC} for odd and even values of $q$.
\begin{example}
	\label{Ex6}
	 Consider \( f(\boldsymbol{u}) = \mathrm{w_L}(\boldsymbol{u}) \), where \( \boldsymbol{u} \in \mathbb{Z}_5^3 \). The expressiveness, \( E = k \left\lfloor \frac{q}{2} \right\rfloor + 1 = 7 \), i.e., \( f(\boldsymbol{u}) \in \{0,1,2,3,4,5,6\} \).
	By \textit{Construction~\ref{Cons1}}, the assigned parity symbols $p(\boldsymbol{u}) = p_{\boldsymbol{u}}$ corresponding to the function values given in the same order above for $t=1$ are \( 0, 2, 4, 1, 3, 0, 2\), respectively. The corresponding parity symbols are listed in Table~\ref{Tab_Ex_FCLC_O}. 
\end{example}

\begin{table*}[!htbp]
	\centering
\caption{FCLC for the Lee Weight function in Example~\ref{Ex6}.}
	\vspace{1mm}
	\setlength{\tabcolsep}{5.4pt}
	\renewcommand{\arraystretch}{1.48}
	\setlength{\arrayrulewidth}{0.2pt} % default is 0.4pt
	\begin{minipage}{0.5\textwidth}
		\centering
		\begin{tabular}{|c|c|c||c|c|c||c|c|c|}
			\hline
			$\mathbf{u}$ & $f(\mathbf{u})$ & $p(\boldsymbol{u})$  & $\mathbf{u}$ &
			$f(\mathbf{u})$ & $p(\boldsymbol{u})$  & $\mathbf{u}$ & $f(\mathbf{u})$  & $p(\boldsymbol{u})$ \\ \hline
			$000$ & 0 & 0 & $041$ & 2 & 4 & $132$ & 5 & 0 \\ \hline
			$001$ & 1 & 2 & $042$ & 3 & 1 & $133$ & 5 & 0 \\ \hline
			$002$ & 2 & 4 & $043$ & 3 & 1 & $134$ & 4 & 3 \\ \hline
			$003$ & 2 & 4 & $044$ & 2 & 4 & $140$ & 2 & 4 \\ \hline
			$004$ & 1 & 2 & $100$ & 1 & 2 & $141$ & 3 & 1 \\ \hline
			$010$ & 1 & 2 & $101$ & 2 & 4 & $142$ & 4 & 3 \\ \hline
			$011$ & 2 & 4 & $102$ & 3 & 1 & $143$ & 4 & 3 \\ \hline
			$012$ & 3 & 1 & $103$ & 3 & 1 & $144$ & 3 & 1 \\ \hline
			$013$ & 3 & 1 & $104$ & 2 & 4 & $200$ & 2 & 4 \\ \hline
			$014$ & 2 & 4 & $110$ & 2 & 4 & $201$ & 3 & 1 \\ \hline
			$020$ & 2 & 4 & $111$ & 3 & 1 & $202$ & 4 & 3 \\ \hline
			$021$ & 3 & 1 & $112$ & 4 & 3 & $203$ & 4 & 3 \\ \hline
			$022$ & 4 & 3 & $113$ & 4 & 3 & $204$ & 3 & 1 \\ \hline
			$023$ & 4 & 3 & $114$ & 3 & 1 & $210$ & 3 & 1 \\ \hline
			$024$ & 3 & 1 & $120$ & 3 & 1 & $211$ & 4 & 3 \\ \hline
			$030$ & 2 & 4 & $121$ & 4 & 3 & $212$ & 5 & 0 \\ \hline
			$031$ & 3 & 1 & $122$ & 5 & 0 & $213$ & 5 & 0 \\ \hline
			$032$ & 4 & 3 & $123$ & 5 & 0 & $214$ & 4 & 3 \\ \hline
			$033$ & 4 & 3 & $124$ & 4 & 3 & $220$ & 4 & 3 \\ \hline
			$034$ & 3 & 1 & $130$ & 3 & 1 & $221$ & 5 & 0 \\ \hline
			$040$ & 1 & 2 & $131$ & 4 & 3 & $222$ & 6 & 2 \\ \hline
		\end{tabular}
	\end{minipage}%
	\hfill
	\begin{minipage}{0.5\textwidth}
		\centering
		\begin{tabular}{|c|c|c||c|c|c||c|c|c|}
			\hline
			$\mathbf{u}$ & $f(\mathbf{u})$ & $p(\boldsymbol{u})$  & $\mathbf{u}$ &
			$f(\mathbf{u})$ & $p(\boldsymbol{u})$  & $\mathbf{u}$ & $f(\mathbf{u})$  & $p(\boldsymbol{u})$ \\ \hline
			$223$ & 6 & 2 & $314$ & 4 & 3 & $410$ & 2 & 4 \\ \hline
			$224$ & 5 & 0 & $320$ & 4 & 3 & $411$ & 3 & 1 \\ \hline
			$230$ & 4 & 3 & $321$ & 5 & 0 & $412$ & 4 & 3 \\ \hline
			$231$ & 5 & 0 & $322$ & 6 & 2 & $413$ & 4 & 3 \\ \hline
			$232$ & 6 & 2 & $323$ & 6 & 2 & $414$ & 3 & 1 \\ \hline
			$233$ & 6 & 2 & $324$ & 5 & 0 & $420$ & 3 & 1 \\ \hline
			$234$ & 5 & 0 & $330$ & 4 & 3 & $421$ & 4 & 3 \\ \hline
			$240$ & 3 & 1 & $331$ & 5 & 0 & $422$ & 5 & 0 \\ \hline
			$241$ & 4 & 3 & $332$ & 6 & 2 & $423$ & 5 & 0 \\ \hline
			$242$ & 5 & 0 & $333$ & 6 & 2 & $424$ & 4 & 3 \\ \hline
			$243$ & 5 & 0 & $334$ & 5 & 0 & $430$ & 3 & 1 \\ \hline
			$244$ & 4 & 3 & $340$ & 3 & 1 & $431$ & 4 & 3 \\ \hline
			$300$ & 2 & 4 & $341$ & 4 & 3 & $432$ & 5 & 0 \\ \hline
			$301$ & 3 & 1 & $342$ & 5 & 0 & $433$ & 5 & 0 \\ \hline
			$302$ & 4 & 3 & $343$ & 5 & 0 & $434$ & 4 & 3 \\ \hline
			$303$ & 4 & 3 & $344$ & 4 & 3 & $440$ & 2 & 4 \\ \hline
			$304$ & 3 & 1 & $400$ & 1 & 2 & $441$ & 3 & 1 \\ \hline
			$310$ & 3 & 1 & $401$ & 2 & 4 & $442$ & 4 & 3 \\ \hline
			$311$ & 4 & 3 & $402$ & 3 & 1 & $443$ & 4 & 3 \\ \hline
			$312$ & 5 & 0 & $403$ & 3 & 1 & $444$ & 3 & 1 \\ \hline
			$313$ & 5 & 0 & $404$ & 2 & 4 &  &  &  \\ \hline
		\end{tabular}
	\end{minipage}
\label{Tab_Ex_FCLC_O}
\end{table*}

%%%%%%%%%%%%
\begin{example} 
	\label{Ex7} Consider \( f(\boldsymbol{u}) = \mathrm{w_L}(\boldsymbol{u}) \), where \( \boldsymbol{u} \in \mathbb{Z}_6^3 \). The expressiveness, \( E = k \left\lfloor \frac{q}{2} \right\rfloor + 1 = 10 \), i.e., \( f(\boldsymbol{u}) \in \{0,1,2,3,4,5,6,7,8,9\} \).
	By \textit{Construction~\ref{Cons1}}, the assigned parity vectors $p(\boldsymbol{u}) = (p_{\boldsymbol{u}}, p_{\boldsymbol{u'}})$ corresponding to the function values given in the same order above for $t=1$ are \( 00,  20, 40, 10, 30, 53, 00,  20, 40, 10\), respectively. The corresponding parity vectors are listed in Table~\ref{Tab_Ex_FCLC_LW_E}.	
\end{example}
\begin{table*}[!htbp]
	\centering
	\caption{FCLC for the Lee Weight function in Example~\ref{Ex7}.}
	\vspace{1mm}
	\setlength{\tabcolsep}{5.5pt}
	\renewcommand{\arraystretch}{1.48}
	\setlength{\arrayrulewidth}{0.2pt} % default is 0.4pt
	\begin{minipage}{0.5\textwidth}
		\centering
		\begin{tabular}{|c|c|c||c|c|c||c|c|c|}
			\hline
			$\mathbf{u}$ & $f(\mathbf{u})$ & $p(\boldsymbol{u})$  & $\mathbf{u}$ &
			$f(\mathbf{u})$ & $p(\boldsymbol{u})$  & $\mathbf{u}$ & $f(\mathbf{u})$  & $p(\boldsymbol{u})$ \\ \hline
			000 & 0 & 00 & 100 & 1 & 20 & 200 & 2 & 40 \\ \hline
			001 & 1 & 20 & 101 & 2 & 40 & 201 & 3 & 10 \\ \hline
			002 & 2 & 40 & 102 & 3 & 10 & 202 & 4 & 30 \\ \hline
			003 & 3 & 10 & 103 & 4 & 30 & 203 & 5 & 53 \\ \hline
			004 & 2 & 40 & 104 & 3 & 10 & 204 & 4 & 30 \\ \hline
			005 & 1 & 20 & 105 & 2 & 40 & 205 & 3 & 10 \\ \hline
			010 & 1 & 20 & 110 & 2 & 40 & 210 & 3 & 10 \\ \hline
			011 & 2 & 40 & 111 & 3 & 10 & 211 & 4 & 30 \\ \hline
			012 & 3 & 10 & 112 & 4 & 30 & 212 & 5 & 53 \\ \hline
			013 & 4 & 30 & 113 & 5 & 53 & 213 & 6 & 00 \\ \hline
			014 & 3 & 10 & 114 & 4 & 30 & 214 & 5 & 53 \\ \hline
			015 & 2 & 40 & 115 & 3 & 10 & 215 & 4 & 30 \\ \hline
			020 & 2 & 40 & 120 & 3 & 10 & 220 & 4 & 30 \\ \hline
			021 & 3 & 10 & 121 & 4 & 30 & 221 & 5 & 53 \\ \hline
			022 & 4 & 30 & 122 & 5 & 53 & 222 & 6 & 00 \\ \hline
			023 & 5 & 53 & 123 & 6 & 00 & 223 & 7 & 20 \\ \hline
			024 & 4 & 30 & 124 & 5 & 53 & 224 & 6 & 00 \\ \hline
			025 & 3 & 10 & 125 & 4 & 30 & 225 & 5 & 53 \\ \hline
			030 & 3 & 10 & 130 & 4 & 30 & 230 & 5 & 53 \\ \hline
			031 & 4 & 30 & 131 & 5 & 53 & 231 & 6 & 00 \\ \hline
			032 & 5 & 53 & 132 & 6 & 00 & 232 & 7 & 20 \\ \hline
			033 & 6 & 00 & 133 & 7 & 20 & 233 & 8 & 40 \\ \hline
			034 & 5 & 53 & 134 & 6 & 00 & 234 & 7 & 20 \\ \hline
			035 & 4 & 30 & 135 & 5 & 53 & 235 & 6 & 00 \\ \hline
			040 & 2 & 40 & 140 & 3 & 10 & 240 & 4 & 30 \\ \hline
			041 & 3 & 10 & 141 & 4 & 30 & 241 & 5 & 53 \\ \hline
			042 & 4 & 30 & 142 & 5 & 53 & 242 & 6 & 00 \\ \hline
			043 & 5 & 53 & 143 & 6 & 00 & 243 & 7 & 20 \\ \hline
			044 & 4 & 30 & 144 & 5 & 53 & 244 & 6 & 00 \\ \hline
			045 & 3 & 10 & 145 & 4 & 30 & 245 & 5 & 53 \\ \hline
			050 & 1 & 20 & 150 & 2 & 40 & 250 & 3 & 10 \\ \hline
			051 & 2 & 40 & 151 & 3 & 10 & 251 & 4 & 30 \\ \hline
			052 & 3 & 10 & 152 & 4 & 30 & 252 & 5 & 53 \\ \hline
			053 & 4 & 30 & 153 & 5 & 53 & 253 & 6 & 00 \\ \hline
			054 & 3 & 10 & 154 & 4 & 30 & 254 & 5 & 53 \\ \hline
			055 & 2 & 40 & 155 & 3 & 10 & 255 & 4 & 30 \\ \hline
		\end{tabular}
	\end{minipage}%
	\hfill
	\begin{minipage}{0.5\textwidth}
		\centering
		\begin{tabular}{|c|c|c||c|c|c||c|c|c|}
			\hline
			$\mathbf{u}$ & $f(\mathbf{u})$ & $p(\boldsymbol{u})$  & $\mathbf{u}$ &
			$f(\mathbf{u})$ & $p(\boldsymbol{u})$  & $\mathbf{u}$ & $f(\mathbf{u})$  & $p(\boldsymbol{u})$ \\ \hline
			300 & 3 & 10 & 400 & 2 & 40 & 500 & 1 & 20 \\ \hline
			301 & 4 & 30 & 401 & 3 & 10 & 501 & 2 & 40 \\ \hline
			302 & 5 & 53 & 402 & 4 & 30 & 502 & 3 & 10 \\ \hline
			303 & 6 & 00 & 403 & 5 & 53 & 503 & 4 & 30 \\ \hline
			304 & 5 & 53 & 404 & 4 & 30 & 504 & 3 & 10 \\ \hline
			305 & 4 & 30 & 405 & 3 & 10 & 505 & 2 & 40 \\ \hline
			310 & 4 & 30 & 410 & 3 & 10 & 510 & 2 & 40 \\ \hline
			311 & 5 & 53 & 411 & 4 & 30 & 511 & 3 & 10 \\ \hline
			312 & 6 & 00 & 412 & 5 & 53 & 512 & 4 & 30 \\ \hline
			313 & 7 & 20 & 413 & 6 & 00 & 513 & 5 & 53 \\ \hline
			314 & 6 & 00 & 414 & 5 & 53 & 514 & 4 & 30 \\ \hline
			315 & 5 & 53 & 415 & 4 & 30 & 515 & 3 & 10 \\ \hline
			320 & 5 & 53 & 420 & 4 & 30 & 520 & 3 & 10 \\ \hline
			321 & 6 & 00 & 421 & 5 & 53 & 521 & 4 & 30 \\ \hline
			322 & 7 & 20 & 422 & 6 & 00 & 522 & 5 & 53 \\ \hline
			323 & 8 & 40 & 423 & 7 & 20 & 523 & 6 & 00 \\ \hline
			324 & 7 & 20 & 424 & 6 & 00 & 524 & 5 & 53 \\ \hline
			325 & 6 & 00 & 425 & 5 & 53 & 525 & 4 & 30 \\ \hline
			330 & 6 & 00 & 430 & 5 & 53 & 530 & 4 & 30 \\ \hline
			331 & 7 & 20 & 431 & 6 & 00 & 531 & 5 & 53 \\ \hline
			332 & 8 & 40 & 432 & 7 & 20 & 532 & 6 & 00 \\ \hline
			333 & 9 & 10 & 433 & 8 & 40 & 533 & 7 & 20 \\ \hline
			334 & 8 & 40 & 434 & 7 & 20 & 534 & 6 & 00 \\ \hline
			335 & 7 & 20 & 435 & 6 & 00 & 535 & 5 & 53 \\ \hline
			340 & 5 & 53 & 440 & 4 & 30 & 540 & 3 & 10 \\ \hline
			341 & 6 & 00 & 441 & 5 & 53 & 541 & 4 & 30 \\ \hline
			342 & 7 & 20 & 442 & 6 & 00 & 542 & 5 & 53 \\ \hline
			343 & 8 & 40 & 443 & 7 & 20 & 543 & 6 & 00 \\ \hline
			344 & 7 & 20 & 444 & 6 & 00 & 544 & 5 & 53 \\ \hline
			345 & 6 & 00 & 445 & 5 & 53 & 545 & 4 & 30 \\ \hline
			350 & 4 & 30 & 450 & 3 & 10 & 550 & 2 & 40 \\ \hline
			351 & 5 & 53 & 451 & 4 & 30 & 551 & 3 & 10 \\ \hline
			352 & 6 & 00 & 452 & 5 & 53 & 552 & 4 & 30 \\ \hline
			353 & 7 & 20 & 453 & 6 & 00 & 553 & 5 & 53 \\ \hline
			354 & 6 & 00 & 454 & 5 & 53 & 554 & 4 & 30 \\ \hline
			355 & 5 & 53 & 455 & 4 & 30 & 555 & 3 & 10 \\ \hline
		\end{tabular}
	\end{minipage}
	\label{Tab_Ex_FCLC_LW_E}
\end{table*}

%%%%%%%%%%
The following example shows that \textit{Construction~\ref{Cons1}} achieves optimal redundancy for the Lee weight function in the cases \( (q, k, t) = (5, 2, 1) \) and \( (q, k, t) = (7, 2, 2) \).
\begin{example}
\textit{Construction~\ref{Cons1}} achieves optimal redundancy for the Lee weight function when \( (q, k, t) = (5, 2, 1) \) and \( (q, k, t) = (7, 2, 2) \). For the given sets of parameters, \textit{Construction~\ref{Cons1}} yields redundancy values of \(  1 \) and \(  2 \), respectively. According to Corollary~\ref{Cor_PB_LW}, the corresponding lower bounds on redundancy for these parameters are \( 0.73 \) and \( 1.19 \), respectively. Since the redundancy achieved by Construction~\ref{Cons1} matches the smallest integer greater than or equal to the theoretical lower bound in both cases, Construction~\ref{Cons1} is optimal for these parameter settings.
\end{example} 

\subsection{Lee Weight Distribution Function}
We now introduce the Lee weight distribution function, an important function that helps characterize the weight distribution of codewords under the Lee metric.
\begin{defn}[Lee Weight Distribution Function]
 A Lee weight distribution function is defined as  \(f(\boldsymbol{u}) = \Delta_T(\boldsymbol{u}) \triangleq \left\lfloor \frac{\mathrm{w_L}(\boldsymbol{u})}{T} \right\rfloor\), where \( \boldsymbol{u} \in \mathbb{Z}_q^k \) and $T, k \in \mathbb{N}$. 
\end{defn}

For simplicity, we assume that \( T \) divides \( k \left\lfloor \frac{q}{2} \right\rfloor + 1 \). Under this condition, the number of distinct function values is given by \( E = \frac{k\left\lfloor \frac{q}{2} \right\rfloor+1}{T} \). This function defines a step threshold function, based on the Lee weight of \( \boldsymbol{u} \), with \( E - 1 \) uniform steps. The function value increments by one at every multiple of $T$. An upper bound on redundancy of FCLC for this function is obtained by giving an explicit FCLC construction. For a binary function, a set of representative information vectors can be explicitly identified for which the distance requirement matrix and the function distance matrix are the same. In this case, a lower bound on redundancy is obtained. \textit{Construction~\ref{Cons1}} can be used to design FCLCs for the Lee weight distribution function also, as shown in the following lemma.
\begin{lem}
	\label{Lem_Cons1_LWD}
	 \textit{Construction~\ref{Cons1}} yields an FCLC for the Lee weight distribution function for $T \ge 2$, under the following conditions.
	\begin{itemize}
		\item \textbf{Odd $q$:} An FCLC with a redundancy of $t$, provided that $q \ge 5$ and $t \le \frac{(q-3)T}{2}-1$.
		\item  \textbf{Even $q$:} An FCLC with a redundancy of $t+1$, provided that $q \ge 6$ and $t \le min(\frac{(q-4)T}{2}-1, \frac{q}{2})$.
	\end{itemize}
\end{lem}
Proof of Lemma \ref{Lem_Cons1_LWD} is given in Appendix \ref{appendix:Lem_Cons1_LWD}. The following example illustrates Lemma \ref{Lem_Cons1_LWD} for the Lee weight distribution function.
\begin{example}
	 Consider \( f(\boldsymbol{u}) = \Delta_T(\boldsymbol{u}) \), where \( \boldsymbol{u} \in \mathbb{Z}_5^7 \) and $T=3$. The expressiveness $E = \frac{k\left\lfloor \frac{q}{2} \right\rfloor+1}{T} = 5$, i.e., \( f(\boldsymbol{u}) \in \{0,1,2,3,4\} \).
	By \textit{Construction~\ref{Cons1}}, the assigned parity symbols corresponding to the function values given in the same order above for $t=1$ are
	\( 0,  2, 4, 1, 3\), respectively.
\end{example}
The expressiveness of the Lee weight distribution function is given by \( E = \frac{k\left\lfloor \frac{q}{2} \right\rfloor+1}{T} \). For a binary Lee weight distribution function, $E$ reduces to $2$, i.e., $ \Delta_T(\boldsymbol{u}) \in \{0,1\}$.  Under this condition, we derive the following lower bound on the redundancy of an FCLC for the Lee weight distribution function as stated in the subsequent corollary.

\begin{cor}
	\label{Cor_PB_LWD}
    For a binary Lee weight distribution function, we have
    \begin{equation*}
    	r_L^\Delta(q, k,t) \ge
    	\left\{
    	\begin{array}{ll}
    		\displaystyle \frac{4qt}{(q^2-1)}, & \text{odd q,} \\[6pt]
    		\displaystyle \frac{4t}{q}, & \text{even q.} \\[6pt]
    	\end{array}
    	\right.
    \end{equation*}
\end{cor}
\begin{IEEEproof}
 A set of representative information vectors exists for a binary Lee weight distribution function. The representative vectors are any two vectors $\boldsymbol{u}_1$ and $\boldsymbol{u}_2 \in \mathbb{Z}_{q}^k$ with Lee weights $T-1$ and $T$, respectively. This implies that \( d_L^f(f_1, f_2) = d_L(\mathbf{u}_1, \mathbf{u}_2) =1 \), thus the function distance matrix and the distance requirement matrix have identical entries for the chosen representatives $\boldsymbol{u}_1$ and $\boldsymbol{u}_2$. That is \(
 \mathbf{D}_f(t, \mathbf{u}_1, \mathbf{u}_2) = \mathbf{D}_f(t, f_1, f_2) =
 \begin{bmatrix}
 	0 & 2t  \\
 	2t & 0  \\
 \end{bmatrix}
 \text{ for any } t\) and $\sum_{i,j: i<j} [\mathbf{D}]_{ij} = 2t$. Therefore, the redundancy is exactly given by $ r_L^\Delta(q, k,t) = N_{L}\left(\mathbf{D}_f(t, f_1, f_2)\right)$ by Corollary \ref{Cor_OP}. The lower bounds on optimal redundancy are then obtained by  the direct simplification of the Plotkin-like bound of Theorem~\ref{Thm_PB} for odd and even values of $q \ge 6$.
	
\end{IEEEproof}
The next example shows one such case of Lee weight distribution function with $E = 2$, where $ r_L^\Delta(q, k,t) = N_{L}\left(\mathbf{D}_f(t, f_1, f_2)\right)$ for $t=1$.
\begin{example}
	For the Lee weight distribution function, $E = 2$ for $(q, k, T) = \{6, 3, 5\} \text{ and } (q, k, T) = \{7, 3, 5\}$. Let $\mathbf{u}_1 = 220$ and $\mathbf{u}_2 = 221$ be the two representative information vectors with function values $4$ and $5$, respectively. Since  \( d_L^f(f_1, f_2) = d_L(\mathbf{u}_1, \mathbf{u}_2) =1 \) for the chosen vectors, the function distance matrix and the distance requirement matrix are identical, and is given by, 
	\[
	\mathbf{D}_f(1, \mathbf{u}_1, \mathbf{u}_2) = \mathbf{D}_f(1, f_1, f_2) =
	\begin{bmatrix}
		0 & 2  \\
		2 & 0  \\
	\end{bmatrix}
	\text{ for } t=1.\]
\end{example}

The following example shows cases where \textit{Construction~\ref{Cons1}} achieves optimal redundancy for the Lee weight distribution function, specifically when $ q=7$ and  $t =1,2$.
\begin{example}
 \textit{Construction~\ref{Cons1}} achieves optimal redundancy for the Lee weight distribution function when $ q=7, k=3, T=5$ for  $t=1,2$. These are cases of Lee weight distribution function with $E = 2$. The redundancy values obtained from Construction~\ref{Cons1} for $ q=7, t=1,2$ are  \( r = 1,2 \), respectively. According to Corollary~\ref{Cor_PB_LWD}, the corresponding theoretical lower bounds are  \( r_L^\mathrm{w}(q, k,t) = 0.58 \text{ and } 1.16 \) for $t=1$ and $2$, respectively. In each case, the redundancy achieved by Construction~\ref{Cons1} matches the smallest integer greater than or equal to the lower bound, hence it is optimal. 
\end{example}
\subsection{Modular Sum Function}
We next define the modular sum function, which computes the modulo- $q$ sum of all components of a given vector. Modular sum functions arise naturally in communication systems where symbols belong to cyclic alphabets. For example, in M-ary Differential Phase Shift Keying (M-DPSK) \cite{JP}, the phase evolution of transmitted symbols is determined by successive phase increments relative to the previous phase, which can be modeled through a modular accumulation process. When the phase values are quantized into $q$ discrete levels over $[0,2\pi)$, they can be represented as elements of the cyclic alphabet $\mathbb{Z}_q$, and the phase evolution can be expressed as a modular sum of the phase increments. Such settings naturally lead to modular sum functions defined over $\mathbb{Z}_q$. Since channel noise typically introduces small cyclic perturbations in phase values, the Lee metric provides a natural error model for such systems \cite{KN}, thereby motivating the study of function-correcting Lee codes for modular sum functions.
\begin{defn}[Modular Sum Function]
	 A Modular sum function is defined as $f(\boldsymbol{u}) = S_{m}(\boldsymbol{u})= (\sum_{i=1}^{k} u_{i})\bmod q$, $\forall$ $\mathbf{u} = (u_{1}, u_{2}, \dots, u_{k}) \in\mathbb{Z}_{q}^{k}$, and  therefore \( E = |Im(f)| = q \), independent of the length $k$ of the information vectors.
\end{defn}
The optimal redundancy $r_L^{\mathrm{S_{m}}}(q,t)$ of a modular sum function is given by Corollary~\ref{Cor_OP} as shown in the next lemma.
\begin{lem}
	\label{Lem_MS_Cor2}
	 Let \( f : \mathbb{Z}_q^k \to \mathbb{Z}_q \) be the modular sum function defined by \( f(\mathbf{u}) = (\sum_{i=1}^k u_i) \mod q.\) Consider the set of \( q \) representative information vectors \( \mathbf{u}_1, \ldots, \mathbf{u}_q \in \mathbb{Z}_q^k \) defined by
	\( \mathbf{u}_i = (0^{k-1}, i-1), \quad \forall i \in [q]. \)
	Then, for this set of representative vectors, the distance requirement matrix and the function distance matrix are identical and \( r_L^{\mathrm{S_{m}}}(q,t) = N_L\left( \mathbf{D}_{\mathrm{S_{m}}}(E,t) \right),\) where \( \mathbf{D}_{\mathrm{S_{m}}}(E,t) \) denotes the function distance matrix for the modular sum function for $t$ error correction.
\end{lem}

\begin{IEEEproof}
	For each representative vector \( \mathbf{u}_i = (0^{k-1}, i-1) \), we have \( f(\mathbf{u}_i) = i-1  \pmod {q}.\)
	Hence, the function values of the representative vectors are \( f(\mathbf{u}_i) = i-1 \), for \( i \in [q] \). Thus, the total number of distinct function values is \( E = q \), and we have \( M = q \) representative information vectors.
	From Definition~\ref{Def_DRM}, the distance requirement matrix for \( t \)-error correction with entries
	\[
	[\!\mathbf{D}_f\!(t,\! \mathbf{u}_1, \dots, \mathbf{u}_M)\!]_{ij}\!=\!
	\begin{cases}
		\begin{aligned}
			&\![2t\! +\!1\!-\!d_{L}\!(\!\mathbf{u}_i,\! \mathbf{u}_j\!)\!]^+\!, \text{if }\! f(\mathbf{u}_i)\! \neq\! f(\mathbf{u}_j), \\
			&0, \quad \text{otherwise}.
		\end{aligned}
	\end{cases}
	\]
	Each representative vector differs only in the last coordinate. Thus, the Lee distance between \( \mathbf{u}_i \) and \( \mathbf{u}_j \) is
	\(
	d_L(\mathbf{u}_i, \mathbf{u}_j) = d_L((0,\ldots,0,i-1), (0,\ldots,0,j-1)) = d_L(i-1, j-1),
	\)
	Therefore, \(\left[\mathbf{D}_f(t, \mathbf{u}_1, \ldots, \mathbf{u}_q)\right]_{ij} = \left[ 2t + 1 - d_L(i-1, j-1) \right]^+.\)
	As per Definition~\ref{Def_FDM}, the function distance matrix for $t$ error correction with entries
	\[
	\left[\mathbf{D}_f(t, f_1, \ldots, f_q)\right]_{ij} =
	\begin{cases}
		\left[ 2t + 1 - d_L^f(f_i, f_j) \right]^+ & \text{if } f_i \ne f_j, \\
		0 & \text{otherwise},
	\end{cases}
	\]
	where \( d_L^f(f_i, f_j) = \min\{ d_L(\mathbf{u}, \mathbf{v}) : f(\mathbf{u}) = f_i,\, f(\mathbf{v}) = f_j \} \). It can be noted that, each function value \( f_i = i-1 \) corresponds to a unique vector \( \mathbf{u}_i \) among the set of representative vectors  with last coordinate \( i-1 \) and zeros elsewhere. Therefore, the pair \( (\mathbf{u}_i, \mathbf{u}_j) \) achieves the minimum possible distance between any pair of vectors whose modular sum values are \( f_i \) and \( f_j \), respectively. Hence, \(	d_L^f(f_i, f_j) = d_L(\mathbf{u}_i, \mathbf{u}_j) = d_L(i-1, j-1).\) Substituting, we get, \( \left[\mathbf{D}_f(t, f_1, \ldots, f_q)\right]_{ij} = \left[ 2t + 1 - d_L(i-1, j-1) \right]^+. \) Since both the distance requirement matrix and the function distance matrix have entries equal to \( \left[ 2t + 1 - d_L(i-1, j-1) \right]^+ \) for \( i \ne j \), and both matrices are of size \( q \times q \), they are identical. Thus, \( \mathbf{D}_f(t, \mathbf{u}_1, \ldots, \mathbf{u}_q) = \mathbf{D}_f(t, f_1, \ldots, f_q).\)
	By Corollary~\ref{Cor_OP}, the optimal redundancy of FCLCs for the modular sum function satisfies
	\( r_L^{\mathrm{S_{m}}}(q, t) = N_L\left( \mathbf{D}_{\mathrm{S_{m}}}(E,t) \right). \)
\end{IEEEproof}
The following example illustrates Lemma~\ref{Lem_MS_Cor2}.
\begin{example} 
	 Consider 
	\(
	\mathbf{u}_1 = 000, \;
	\mathbf{u}_2 = 001, \;
	\mathbf{u}_3 = 002, \;
	\mathbf{u}_4 = 003, \;
	\mathbf{u}_5 = 004, \;
	\mathbf{u}_6 = 005
	\) over \( \mathbb{Z}_6 \).
	The function values \( f(\mathbf{u}_i) = i-1 \),  $\forall$ \( i \in [6]\) and \(Im(f) = \{0, 1, 2, 3 ,4 ,5\}\). The Lee distance between any pair
	\(
	d_L(\mathbf{u}_i, \mathbf{u}_j) = d_L(i-1, j-1).
	\) For the chosen representative vectors $\{\mathbf{u}_1, \mathbf{u}_2, \mathbf{u}_3, \mathbf{u}_4, \mathbf{u}_5, \mathbf{u}_6\}$ and \( t = 1 \), the distance requirement matrix is given by,
	\[
	\mathbf{D}_f(t=1, \mathbf{u}_1, \ldots, \mathbf{u}_6) =
	\begin{bmatrix}
		0 & 2 & 1 & 0 & 1 & 2 \\
		2 & 0 & 2 & 1 & 0 & 1 \\
		1 & 2 & 0 & 2 & 1 & 0 \\
		0 & 1 & 2 & 0 & 2 & 1 \\
		1 & 0 & 1 & 2 & 0 & 2 \\
		2 & 1 & 0 & 1 & 2 & 0 \\
	\end{bmatrix}
	\]
	The function distance matrix for $\{f(\mathbf{u}_1) =f_1, \ldots, f(\mathbf{u}_6)=f_6\}$ and $t=1$ is given by, 
	\[
	\mathbf{D}_f(t=1, f_1, \ldots, f_6) =
	\begin{bmatrix}
		0 & 2 & 1 & 0 & 1 & 2 \\
		2 & 0 & 2 & 1 & 0 & 1 \\
		1 & 2 & 0 & 2 & 1 & 0 \\
		0 & 1 & 2 & 0 & 2 & 1 \\
		1 & 0 & 1 & 2 & 0 & 2 \\
		2 & 1 & 0 & 1 & 2 & 0 \\
	\end{bmatrix}
	\]
	The two matrices are identical, thereby validating Lemma~\ref{Lem_MS_Cor2}. Hence,
	\(
	r_L^{\mathrm{S_{m}}}(6, 1) = N_L\left( \mathbf{D}_{\mathrm{S_{m}}}(\mathrm{6},1) \right).
	\)
\end{example}

A construction is presented below for designing FCLCs for the modular sum function. 
\begin{cons}
	\label{Cons2_MS}
	 For \( q \ge 5, \boldsymbol{u} \in \mathbb{Z}_q^k \) and a function $f : \mathbb{Z}_q^k \rightarrow \mathbb{Z}_q$. The encoding of $\boldsymbol{u}$ is given by
	$ \mathrm{Enc}(\boldsymbol{u}) = (\boldsymbol{u}, \{p_{\boldsymbol{u}}\}^{t})$ when $q$ is odd, and $ \mathrm{Enc}(\boldsymbol{u}) = (\boldsymbol{u}, \{p_{\boldsymbol{u}}\}^{t},p'_{\boldsymbol{u}} )$ when $q$ is even, where
	\begin{equation}
		\label{eq6}
		p_{\boldsymbol{u}} =
		\left\{
		\begin{array}{ll}
			\displaystyle (2f(\boldsymbol{u})) \bmod q, & \text{if }  \text{ q is odd, }\\[6pt]
			\displaystyle (2f(\boldsymbol{u})) \bmod q, & \text{if } 0 \le f(\boldsymbol{u})  \le  \left(  \frac{q}{2}  - 1 \right) \\& \text{and q is even, } \\[6pt]
			\displaystyle (2f(\boldsymbol{u}) + 1) \bmod q, & \text{if }  \frac{q}{2}  \le f(\boldsymbol{u})  \le (q - 1) \\&\text{ and q is even, } 
		\end{array}
		\right.
	\end{equation}
	\begin{equation}
		\label{eq6b}
		%\scriptsize
		p'_{\boldsymbol{u}} =
		\left\{
		\begin{array}{ll}
			\displaystyle \frac{q}{2}, & \text{if }  f(\boldsymbol{u})  = (q-1),\\[6pt]
			\displaystyle 0, & \text{otherwise, } 
		\end{array}
		\right.
	\end{equation}
	and \(\{p_{\boldsymbol{u}}\}^{t}\) stands for the \( t \)-fold repetition of the parity symbol \(p_{\boldsymbol{u}}\). 
\end{cons}
\textit{Construction~\ref{Cons2_MS}}  can be used to design FCLCs for the modular sum function as demonstrated in the next lemma.
\begin{lem}
	\label{Lem_Cons2_FCLC}
	 \textit{Construction~\ref{Cons2_MS}} yields an FCLC for the modular sum function, achieving a redundancy of $t$ when $q$ is odd  and redundancy $t+1$ when $q$ is even, under the conditions $q \ge 5$ and $t \le \frac{q - 3}{2}$.
\end{lem}
Proof of Lemma \ref{Lem_Cons2_FCLC} is given in Appendix \ref{appendix:Lem_Cons2_FCLC}. The following two examples illustrate Lemma \ref{Lem_Cons2_FCLC} for odd and even values of $q$.
\begin{example}
	\label{Ex9}
	 Consider \( f(\boldsymbol{u}) = S_{m}(\mathbf{u}) = (\sum_{i=1}^{k} \boldsymbol{u}_{i}) \bmod q\), where \( \mathbf{u} \in \mathbb{Z}_5^2 \). The function takes values in the range  \( f(\mathbf{u}) \in \{ 0,1,2,3,4 \} \), and hence the number of distinct function values, \( E = q = 5 \). To construct an FCLC using \textit{Construction~\ref{Cons2_MS}} for \( t = 1 \), we assign distinct parity symbols $p(\boldsymbol{u}) = p_{\boldsymbol{u}}$ = $0, 2, 4, 1, 3$ to each function value \( i \in \{0, 1, 2, 3, 4\} \), respectively. The corresponding parity symbols are listed in Table~\ref{Tab_Ex_FCLC_S_{m}_O} with redundancy $1$. 
	
\begin{table}[!htbp]
	\centering
	\tiny
	\caption{FCLC for the Modular Sum function in Example~\ref{Ex9} for $t=1$.}
	\vspace{1mm}
	\setlength{\tabcolsep}{4pt}
	\renewcommand{\arraystretch}{1.2}
	\setlength{\arrayrulewidth}{0.2pt} % default is 0.4pt
	\resizebox{0.37\textwidth}{!}{%
		\begin{tabular}{|c|c|c||c|c|c|}
			\hline
			$\boldsymbol{u}$ & $f(\boldsymbol{u})$ & $p(\boldsymbol{u})$ &
			$\boldsymbol{u}$ & $f(\boldsymbol{u})$ & $p(\boldsymbol{u})$ \\
			\hline
			00 & 0 & 0 & 23 & 0 & 0 \\ \hline
			01 & 1 & 2 & 24 & 1 & 2 \\ \hline
			02 & 2 & 4 & 30 & 3 & 1 \\ \hline
			03 & 3 & 1 & 31 & 4 & 3 \\ \hline
			04 & 4 & 3 & 32 & 0 & 0 \\ \hline
			10 & 1 & 2 & 33 & 1 & 2 \\ \hline
			11 & 2 & 4 & 34 & 2 & 4 \\ \hline
			12 & 3 & 1 & 40 & 4 & 3 \\ \hline
			13 & 4 & 3 & 41 & 0 & 0 \\ \hline
			14 & 0 & 0 & 42 & 1 & 2 \\ \hline
			20 & 2 & 4 & 43 & 2 & 4 \\ \hline
			21 & 3 & 1 & 44 & 3 & 1 \\ \hline
			22 & 4 & 3 &        &   &   \\
			\hline
		\end{tabular}
	}
	\label{Tab_Ex_FCLC_S_{m}_O}
\end{table}

\end{example}
\begin{example}
	\label{Ex10}
	Consider \( f(\boldsymbol{u}) = S_{m}(\mathbf{u}) =(\sum_{i=1}^{k} \boldsymbol{u}_{i}) \) mod $q$, where \( \mathbf{u} \in \mathbb{Z}_6^2 \). The function takes values in the range  \( f(\mathbf{u}) \in \{ 0,1,2,3,4,5 \} \), and hence the number of distinct function values, \( E = q = 6 \). FCLC can be designed using \textit{Construction~\ref{Cons2_MS}} for $t=1$ with the following  parity vectors $p(\boldsymbol{u}) = (p_{\boldsymbol{u}}, p_{\boldsymbol{u'}})$ = \(00, 20, 40, 10, 30, 53  \) assigned to each function values \( i \in \{0, 1, 2, 3, 4, 5\} \), respectively. The corresponding parity vectors are listed  in Table~\ref{Tab_Ex_MS_FCLC_E} with redundancy $2$.
\begin{table}[!htbp]
	\centering
	\tiny
	\caption{FCLC for the Modular Sum function in Example~\ref{Ex10} for $t = 1$.}
	\vspace{1mm}
	\setlength{\tabcolsep}{4pt}
	\renewcommand{\arraystretch}{1.2}
	\setlength{\arrayrulewidth}{0.2pt} % default is 0.4pt
	\resizebox{0.37\textwidth}{!}{%
		\begin{tabular}{|c|c|c||c|c|c|}
			\hline
			$\boldsymbol{u}$ & $f(\boldsymbol{u})$ & $p(\boldsymbol{u})$ &
			$\boldsymbol{u}$ & $f(\boldsymbol{u})$ & $p(\boldsymbol{u})$ \\
			\hline
			00 & 0 & 00 & 30 & 3 & 10 \\ \hline
			01 & 1 & 20 & 31 & 4 & 30 \\ \hline
			02 & 2 & 40 & 32 & 5 & 53 \\ \hline
			03 & 3 & 10 & 33 & 0 & 00 \\ \hline
			04 & 4 & 30 & 34 & 1 & 20 \\ \hline
			05 & 5 & 53 & 35 & 2 & 40 \\ \hline
			10 & 1 & 20 & 40 & 4 & 30 \\ \hline
			11 & 2 & 40 & 41 & 5 & 53 \\ \hline
			12 & 3 & 10 & 42 & 0 & 00 \\ \hline
			13 & 4 & 30 & 43 & 1 & 20 \\ \hline
			14 & 5 & 53 & 44 & 2 & 40 \\ \hline
			15 & 0 & 00 & 45 & 3 & 10 \\ \hline
			20 & 2 & 40 & 50 & 5 & 53 \\ \hline
			21 & 3 & 10 & 51 & 0 & 00 \\ \hline
			22 & 4 & 30 & 52 & 1 & 20 \\ \hline
			23 & 5 & 53 & 53 & 2 & 40 \\ \hline
			24 & 0 & 00 & 54 & 3 & 10 \\ \hline
			25 & 1 & 20 & 55 & 4 & 30 \\
			\hline
		\end{tabular}
	}
	\label{Tab_Ex_MS_FCLC_E}
\end{table}

\end{example}
Based on Lemma~\ref{Lem_MS_Cor2}, we derive lower bounds on the
optimal redundancy $r_L^{S_{m}}(q, t)$, by applying the Plotkin-like bound from  Theorem~\ref{Thm_PB} as stated next in Corollary~\ref{Cor_PB_MS_O} and Corollary~\ref{Cor_PB_MS_E}.
\begin{cor}
	\label{Cor_PB_MS_O}
	 For odd $q \ge 5$ and \( t \ge\frac{\left\lfloor \frac{q}{2} \right\rfloor-1}{2} \), we have
	\( r_L^{S_{m}}(q, t) \ge \frac{2}{q(q+1)} (4t+1-\left\lfloor \frac{q}{2} \right\rfloor)(2q-2\left\lfloor \frac{q}{2} \right\rfloor-1).\) 
\end{cor}

\begin{IEEEproof}
	By Lemma~\ref{Lem_MS_Cor2},  for the modular sum function, we have \( r_L^{\mathrm{S_{m}}}(q, t) = N_L\left( \mathbf{D}_{\mathrm{S_{m}}}(E,t) \right)\). The corresponding function distance matrix $\mathbf{D}_{\mathrm{S_{m}}}(E,t)$  for \( t \ge\frac{\left\lfloor \frac{q}{2} \right\rfloor-1}{2} \) is shown in Figure~\ref{Fig_FDM_MS_O}. In this case, the sum of the entries above the main diagonal is given by, $\sum_{i,j: i<j} [\mathbf{D}]_{ij} = (q-1)(2t+1-1) + (q-2)(2t+1-2) + \ldots + (q- \left\lfloor \frac{q}{2} \right\rfloor ) (2t+1 - \left\lfloor \frac{q}{2} \right\rfloor) + (q- (\left\lfloor \frac{q}{2} \right\rfloor+1))(2t+1 - \left\lfloor \frac{q}{2} \right\rfloor)+ \ldots (1)(2t+1-1)$. This expression simplifies to $\sum_{i,j: i<j} [\mathbf{D}]_{ij} = \frac{\left\lfloor \frac{q}{2} \right\rfloor}{2}(4t+1-\left\lfloor \frac{q}{2} \right\rfloor)(2q-2\left\lfloor \frac{q}{2} \right\rfloor-1)$. Substituting this expression into (\ref{Eq_PB1}) yields the Plotkin-like bound for odd $q \ge 5$ and \( t \ge\frac{\left\lfloor \frac{q}{2} \right\rfloor-1}{2} \). 
\end{IEEEproof}
\begin{figure*}[!htbp]
	\centering
	\scriptsize % compact font
	\resizebox{1.0\textwidth}{!}{% scale to two-column width
		$\begin{blockarray}{c@{\hskip 4pt} *{8}{c}}
			& 0 & 1 & \cdots & \lfloor \frac{q}{2}\!\rfloor & (\lfloor \frac{q}{2}\!\rfloor +1) & \cdots & (q-2) & (q-1)  \\
			\begin{block}{c@{\hskip 4pt}[*{8}{c}]}
				0 & 0 & (2t+1-1)  & \cdots & (2t+1-\left\lfloor \frac{q}{2} \right\rfloor) & (2t+1-\left\lfloor \frac{q}{2} \right\rfloor) &  \cdots & (2t+1-2) & (2t+1-1) \\ [3pt]
				1 & (2t+1-1) & 0  & \cdots & (2t+1-(\left\lfloor \frac{q}{2} \right\rfloor-1))   &(2t+1-\left\lfloor \frac{q}{2} \right\rfloor) & \cdots &(2t+1-3) & (2t+1-2) \\ [3pt]
				\vdots & \vdots & \vdots  & \ddots & \ddots &\ddots & \cdots  &\vdots & \vdots \\ [3pt]
				(\lfloor \frac{q}{2}\!\rfloor) & (2t+1-\left\lfloor \frac{q}{2} \right\rfloor) & (2t+1-(\left\lfloor \frac{q}{2} \right\rfloor-1))  & \ddots & \ddots &\ddots & \cdots  & (2t+1-(\left\lfloor \frac{q}{2} \right\rfloor-1)) & (2t+1-\left\lfloor \frac{q}{2} \right\rfloor) \\ [3pt]
				(\lfloor \frac{q}{2}\!\rfloor +1) & (2t+1-\left\lfloor \frac{q}{2} \right\rfloor) & (2t+1-\left\lfloor \frac{q}{2} \right\rfloor)  & \ddots & \ddots &\ddots & \cdots & (2t+1-(\left\lfloor \frac{q}{2} \right\rfloor-2)) & (2t+1-(\left\lfloor \frac{q}{2} \right\rfloor-1)) \\ [3pt]
				\vdots & \vdots & \vdots  & \ddots & \ddots &\ddots & \cdots & \vdots & \vdots \\ [3pt]
				(q-2) & (2t+1-2) & (2t+1-3)  &\cdots & (2t+1-(\left\lfloor \frac{q}{2} \right\rfloor-1)) & (2t+1-(\left\lfloor \frac{q}{2} \right\rfloor-2)) & \cdots &0 & (2t+1-1) \\ [3pt]
				(q-1) & (2t+1-1) & (2t+1-2)   & \cdots & (2t+1-\left\lfloor \frac{q}{2} \right\rfloor)  &  (2t+1-(\left\lfloor \frac{q}{2} \right\rfloor-1)) & \cdots  & (2t+1-1) & 0 \\
			\end{block}
		\end{blockarray}$%
	}
	\caption{Function distance matrix $\mathbf{D}_{\mathrm{S_{m}}}(E,t)$ for $t \ge\frac{\left\lfloor \frac{q}{2} \right\rfloor-1}{2}$ and odd $q \ge 5$.}
	\label{Fig_FDM_MS_O}
\end{figure*}

The following example shows cases where \textit{Construction~\ref{Cons2_MS}} achieves optimal redundancy for the modular sum function, specifically when \( q = 5,\ t = 1,\ \forall k \quad \text{and} \quad q = 7,\ t = 2,\ \forall k. \)
\begin{example}
\textit{Construction~\ref{Cons2_MS}} achieves optimal redundancy for the modular sum function in the following cases:
	\(q = 5,\ t = 1,\ \forall k \quad \text{and} \quad q = 7,\ t = 2,\ \forall k. \) For the case \( q = 5 \), \( t = 1 \), \textit{Construction~\ref{Cons2_MS}} yields a redundancy of \( r = 1 \) \( \forall k \). Similarly, for \( q = 7 \), \( t = 2 \), Construction~\ref{Cons2_MS} yields \( r = 2 \) \( \forall k \). According to Corollary~\ref{Cor_PB_MS_O}, the lower bounds on redundancy for these cases are  \( 1 \) and \( 1.5 \). Since the actual redundancies obtained using \textit{Construction~\ref{Cons2_MS}} match the integer ceilings of these lower bounds, the construction is optimal in both cases. Thus, \textit{Construction~\ref{Cons2_MS}} achieves the minimum possible redundancy for the specified values of \( q \), \( t \) and \( k \).
\end{example}

\begin{cor}
	\label{Cor_PB_MS_E}
	 For even $q \ge 6$ and \( t \ge \frac{\frac{q}{2}-1}{2} \), we have
	$ r_L^{S_{m}}(q, t) \ge \frac{8}{{q}^3} (A+B),$
	where \(A = \sum_{a=1}^{q/2}(q-a)(2t+1-a) \) and 
	\( B = \sum_{b=1}^{(q/2)-1}(q-(q/2+b))(2t+1-(q/2-b). \)
\end{cor}
\begin{IEEEproof}
By Lemma~\ref{Lem_MS_Cor2},  for the modular sum function, we have \( r_L^{\mathrm{S_{m}}}(q, t) = N_L\left( \mathbf{D}_{\mathrm{S_{m}}}(E,t) \right)\). The corresponding function distance matrix $\mathbf{D}_{\mathrm{S_{m}}}(E,t)$  for \( t \ge\frac{ \frac{q}{2} -1}{2} \) is shown in Figure~\ref{Fig_FDM_MS_E}. In this case, the sum of the entries above the main diagonal is given by, $\sum_{i,j: i<j} [\mathbf{D}]_{ij} = (q-1)(2t+1-1) + (q-2)(2t+1-2) + \ldots + (q- (\frac{q}{2}-1)) (2t+1 - ( \frac{q}{2} -1)) + (q-  \frac{q}{2} )(2t+1 -  \frac{q}{2}) + (q- ( \frac{q}{2} +1))(2t+1 - (\frac{q}{2} -1))+ \ldots + (1)(2t+1-1)$. Substituting this expression into (\ref{Eq_PB1}) and simplifying yields the Plotkin-like bound for even $q \ge 6$ and \( t \ge\frac{ \frac{q}{2} -1}{2} \). 
\end{IEEEproof}
\begin{figure*}[!htbp]
	\centering
	\scriptsize % compact font
	\resizebox{1.0\textwidth}{!}{% scale to two-column width
		$\begin{blockarray}{c@{\hskip 4pt} *{9}{c}}
			& 0 & 1 & \cdots & (\tfrac{q}{2} -1) & \tfrac{q}{2} & (\tfrac{q}{2} +1) & \cdots & (q-2) & (q-1)  \\
			\begin{block}{c@{\hskip 4pt}[*{9}{c}]}
				0 & 0 & (2t+1-1) & \cdots & (2t+1-(\tfrac{q}{2} -1))  & (2t+1-\tfrac{q}{2}) & (2t+1-(\tfrac{q}{2} -1)) & \cdots & (2t+1-2) & (2t+1-1) \\[3pt]
				1 & (2t+1-1) & 0 & \cdots & (2t+1-(\tfrac{q}{2} -2))  & (2t+1-(\tfrac{q}{2} -1)) & (2t+1-\tfrac{q}{2}) & \cdots & (2t+1-3)  & (2t+1-2) \\[3pt]
				\vdots & \vdots & \vdots  & \ddots & \ddots & \ddots & \cdots & \vdots & \vdots & \vdots \\[3pt]
				(\tfrac{q}{2} -1) & (2t+1-(\tfrac{q}{2} -1)) & \vdots  & \ddots & \ddots & \ddots & \ddots & \cdots  & (2t+1-(\tfrac{q}{2} -1)) & (2t+1-\tfrac{q}{2}) \\[3pt]
				\tfrac{q}{2} & (2t+1-\tfrac{q}{2}) & \vdots  & \ddots & \ddots & \ddots & \ddots & \cdots  & (2t+1-(\tfrac{q}{2} -2)) & (2t+1-(\tfrac{q}{2} -1)) \\[3pt]
				(\tfrac{q}{2} +1) & (2t+1-(\tfrac{q}{2} -1)) & \vdots  & \ddots & \ddots & \ddots & \ddots & \cdots  &(2t+1-1) & (2t+1-(\tfrac{q}{2} -2)) \\[3pt]
				\vdots & \vdots & \vdots  & \ddots & \ddots & \ddots & \cdots & \vdots & \vdots  & \vdots \\[3pt]
				(q-2) & (2t+1-2) & (2t+1-3)  & \cdots & (2t+1-(\tfrac{q}{2} -1)) & (2t+1-(\tfrac{q}{2} -2)) & (2t+1-1) & \cdots & 0 & (2t+1-1) \\[3pt]
				(q-1) & (2t+1-1) & (2t+1-2)  & \cdots & (2t+1-\tfrac{q}{2}) & (2t+1-(\tfrac{q}{2} -1)) & (2t+1-(\tfrac{q}{2} -2)) & \cdots & (2t+1-1) & 0 \\
			\end{block}
		\end{blockarray}$%
	}
	\caption{Function distance matrix $\mathbf{D}_{\mathrm{S_{m}}}(E,t)$ for $t \ge \tfrac{\tfrac{q}{2} -1}{2}$ and even $q \ge 6$.}
	\label{Fig_FDM_MS_E}
\end{figure*}
When $q$ is sufficiently large relative to $t$, the distance requirements can be satisfied using a single parity symbol, making it possible to construct FCLCs with optimal redundancy equal to one. The following conjecture formalizes this observation.\\

\begin{conj}
	For all integers $t \ge 1$ and all $q \ge 2t(t+1)+1$, there exist a $p \in [2t:\lfloor \frac{q}{2} \rfloor]$, such that the encoding $ \mathrm{Enc}(\boldsymbol{u}) = (\boldsymbol{u}, (pf(\boldsymbol{u})) \bmod q)$ yields an FCLC for the Lee weight function, Lee weight distribution function and the modular sum function achieving an optimal redundancy $1$. For $q = 2t(t+1)+1$, we have $p=2t+1$.
\end{conj}
\subsection{Locally Bounded Function}
A class of functions that assume only a limited number of values within a given Hamming ball, known as locally-bounded functions was introduced in \cite{CR}. Recently, these functions were extended to the Lee metric and studied in \cite{GA}. We propose the construction of FCLCs for the locally-bounded functions, which attain optimal redundancy in some cases.   
\begin{defn}[Function Ball \cite{GA}]
	 The function ball of a function $f : \mathbb{Z}_q^k \to Im(f)$ with radius $\rho$ around $u$ is defined as $B_L^f(\mathbf{u}, \rho) = \{ f(\mathbf{v})\;|\; \mathbf{v} \in \mathbb{Z}_q^k \text{ and } d_L(\mathbf{u}, \mathbf{v}) \le \rho \}.$
\end{defn}
\begin{defn}[Locally bounded function\cite{GA}]
	 A function $f:\mathbb{Z}_q^k\to Im(f)$ is said to be a locally $(\rho,\lambda)_L$-bounded function if $\mid B_L^f(\mathbf{u},\rho)\mid\leq \lambda$, $\forall \mathbf{u}\in \mathbb{Z}_q^k$.
\end{defn}

The following lemma, which extends Lemma~1 from \cite{CR} to the Lee metric setting, is used in the construction of FCLCs for locally-bounded functions. 
\begin{lem}[\cite{GA}]
	\label{Cont_LB}
	 Let $f:\mathbb{Z}_q^k\to Im(f)$ be a locally $(\rho,\lambda)_L$-bounded function. If there is total order on the image set $Im(f)$ such that  $B_L^f(\mathbf{u}, \rho)$ form a contiguous block for all $\mathbf{u}\in \mathbb{Z}_q^k$. Then there exists a mapping $ Col_f: \mathbb{Z}_q^k \to [\lambda] $ such that $ Col_f(\mathbf{u}) \neq Col_f(\mathbf{v})$ for any $ \mathbf{u}, \mathbf{v} \in \mathbb{Z}_q^k $ with $f(\mathbf{u}) \neq f(\mathbf{v})$ and $d_L(\mathbf{u}, \mathbf{v}) \le \rho$.
\end{lem} 

We now present a construction to design FCLCs corresponding to locally $(2t, \lambda)_L$-bounded functions that satisfy the contiguous block condition in Lemma~\ref{Cont_LB}.
\begin{cons}
	\label{Cons4} 
	 Let $f\!\!:\mathbb{Z}_{q}^k\to \!\! Im(f) = \{f_1, f_2, \ldots, f_\lambda, \ldots, f_\mu\}$ be a locally $(2t, \lambda)_L$-bounded function that satisfy the contiguous block condition in Lemma~\ref{Cont_LB} such that $\lambda \le \frac{q}{2}$ and $t$ be the no. of errors in the channel, where 
	$\mu \ge \lambda$ and $t \in \mathbb{N} $.  Define an encoding for $\mathbf{u}$ as $\mathrm{Enc}(\boldsymbol{u}) = (\boldsymbol{u}, \{\mathbf{p}_{Col_f(\boldsymbol{u})}\}^r)$, where 
	\begin{equation*}
		\footnotesize
		\mathbf{p}_{Col_f(\boldsymbol{u})} =
		\left\{
		\begin{array}{ll}
			\displaystyle 0 & \text{if }  Col_f(\boldsymbol{u}) = 1, \\[6pt]
			\displaystyle 2 \left \lfloor \frac{q}{2\lambda} \right \rfloor & \text{if }  Col_f(\boldsymbol{u}) = 2, \\[6pt]
			\displaystyle . & \ldots, \\[6pt]
			\displaystyle 2(i-1)\left \lfloor \frac{q}{2\lambda} \right \rfloor & \text{if } Col_f(\boldsymbol{u}) = i, \\[6pt]
			\displaystyle . & \ldots, \\[6pt]
			\displaystyle 2(\lambda-1)\left \lfloor \frac{q}{2\lambda} \right \rfloor & \text{if } Col_f(\boldsymbol{u}) = \lambda,
		\end{array}
		\right.
	\end{equation*}
	$i \in [\lambda]$, $r = {\left \lceil \frac{t}{\left \lfloor \frac{q}{2\lambda} \right \rfloor} \right \rceil}$ and \(\{\mathbf{p}_{Col_f(\boldsymbol{u})}\}^r\) means \( r \)-fold repetition of the parity symbol $\mathbf{p}_{Col_f(\boldsymbol{u})}$.
\end{cons}
\textit{Construction~\ref{Cons4}} can be used to design FCLCs for the  locally $(2t, \lambda)_L$-bounded function that satisfy the contiguous block condition in Lemma~\ref{Cont_LB} with a redundancy of $r = \left \lceil \frac{t}{\left \lfloor \frac{q}{2\lambda} \right \rfloor} \right \rceil $ as shown in the next lemma.
\begin{lem}
	 Construction~\ref{Cons4} gives an FCLC for the  locally $(2t, \lambda)_L$-bounded function that satisfy the contiguous block condition in Lemma~\ref{Cont_LB} with a redundancy of $r = \left \lceil \frac{t}{\left \lfloor \frac{q}{2\lambda} \right \rfloor} \right \rceil $ for any $\lambda \le \frac{q}{2}$.
\end{lem}
\begin{proof}
	Let $\mathbf{u_i}, \mathbf{u_j} \in \mathbb{Z}_q^k$ be such that $f(\mathbf{u_i}) \neq f(\mathbf{u_j})$, $\forall$ $i, j \in[\lambda]$ and $i \neq j$ for any $\lambda \le \frac{q}{2}$. We have $d_{L}(\mathrm{Enc}(\boldsymbol{u_i}), \mathrm{Enc}(\boldsymbol{u_j})) = d_{L}(\boldsymbol{u_i}, \boldsymbol{u_j}) + d_{L}(\{\mathbf{p}_{Col_f(\boldsymbol{u_i})}\}^r, \{\mathbf{p}_{Col_f(\boldsymbol{u_j})}\}^r)$, where $\mathbf{p}_{Col_f(\boldsymbol{u_i})}$ and $\mathbf{p}_{Col_f(\boldsymbol{u_j})}$ are the parity symbols corresponding to $Col_f(\boldsymbol{u}) = i$ and $Col_f(\boldsymbol{u}) = j$. We consider two cases, Case 1: If $d_L(\mathbf{u_i},\mathbf{u_j})\geq 2t+1$, then  $d_{L}(\mathrm{Enc}(\boldsymbol{u_i}), \mathrm{Enc}(\boldsymbol{u_j})) \ge d_L(\mathbf{u_i},\mathbf{u_j}) \ge 2t+1$. Case 2: If  $d_L(\mathbf{u_i},\mathbf{u_j}) \le 2t$, then $f(\mathbf{u_j})\in B_L^f(\mathbf{u_i},2t)$, and  $f(\mathbf{u_i}) \neq f(\mathbf{u_j}) $. Therefore, $d_{L}(\mathbf{p}_{Col_f(\boldsymbol{u_i})}, \mathbf{p}_{Col_f(\boldsymbol{u_j})}) \ge 2\left \lfloor \frac{q}{2\lambda} \right \rfloor$ from \textit{Construction~\ref{Cons4}} and $d_{L}(\{\mathbf{p}_{Col_f(\boldsymbol{u_i})}\}^r, \{\mathbf{p}_{Col_f(\boldsymbol{u_j})}\}^r) \ge 2 \left \lfloor \frac{q}{2\lambda} \right \rfloor r \ge  2t$. Since $\mathbf{u_i} \neq \mathbf{u_j}$, we have $d_L(\mathbf{u_i},\mathbf{u_j})\geq 1$ and  $d_{L}(\mathrm{Enc}(\boldsymbol{u_i}), \mathrm{Enc}(\boldsymbol{u_j})) = d_{L}(\boldsymbol{u_i}, \boldsymbol{u_j}) + d_{L}(\{\mathbf{p}_{Col_f(\boldsymbol{u_i})}\}^r, \{\mathbf{p}_{Col_f(\boldsymbol{u_j})}\}^r) \ge 2t+1$.
\end{proof}
The next example gives an illustration of \textit{Construction~\ref{Cons4}} for a locally $(2t, \lambda)_L$-bounded function.
\begin{example}
	 Let \( f: \mathbb{Z}_6^k \to \mathrm{Im}(f) \) be a locally \( (2, 3)_L \)-bounded function that satisfy the contiguous block condition in Lemma~\ref{Cont_LB} with parameters \( t = 1 \) and \( \lambda = 3 \), where \( |\mathrm{Im}(f)| \ge 3 \). Consider vectors \( \boldsymbol{u}_1, \boldsymbol{u}_2, \boldsymbol{u}_3 \in \mathbb{Z}_6^k \) such that \( f(\boldsymbol{u}_i) \neq f(\boldsymbol{u}_j) \) and \( d_L(\boldsymbol{u}_i, \boldsymbol{u}_j) \le 2 \) for all distinct \( i, j \in [3] \). The three parity symbols used in the construction are \( \mathbf{p}_1 = 0 \), \( \mathbf{p}_2 = 2 \), and \( \mathbf{p}_3 = 4 \), corresponding to \( \mathrm{Col}_f(\boldsymbol{u}) = 1, 2, 3 \), respectively. Since \( \lambda = \frac{q}{2} \), Construction~\ref{Cons4} guarantees that \(d_L(\mathrm{Enc}(\boldsymbol{u}_i), \mathrm{Enc}(\boldsymbol{u}_j)) = d_L(\boldsymbol{u}_i, \boldsymbol{u}_j) + d_L(\{ \mathbf{p}_{\mathrm{Col}_f(\boldsymbol{u}_i)} \}, \{ \mathbf{p}_{\mathrm{Col}_f(\boldsymbol{u}_j)} \}) \ge 2t + 1=3,\) $\forall$ \( i \ne j \in [3] \), as $d_{L}(\boldsymbol{u_i}, \boldsymbol{u_j}) = 1$ at least.
\end{example}
We now present a lemma that proves the optimality of \textit{Construction~\ref{Cons4}} for certain locally$(2t, 3)_L$-bounded functions.
\begin{lem}
	 Construction~\ref{Cons4} achieves optimal redundancy of $r_L^l(q, k,t) =t$  for a locally $(2t, 3)_L$-bounded function $f$ that satisfy the contiguous block condition in Lemma~\ref{Cont_LB} with $|Im(f)| \ge 3$, if there exist $\mathbf{u}_1, \mathbf{u}_2, \mathbf{u}_3 \in \mathbb{Z}_6^k$ with $f(\mathbf{u}_i) \neq f(\mathbf{u}_j)$ for $i, j \in [3], i \neq j$, such that $d_L(\mathbf{u}_1, \mathbf{u}_2) = 1$, $d_L(\mathbf{u}_1, \mathbf{u}_3) = 1$ and $d_L(\mathbf{u}_2, \mathbf{u}_3) = 2$.
\end{lem}

\begin{proof}
	For a locally $(2t, 3)_L$-bounded function $f: \mathbb{Z}_6^k \to Im(f)$ with $|Im(f)| \ge 3$, Construction~\ref{Cons4} gives a redundancy of $r = \left \lceil \frac{t}{\left \lfloor \frac{q}{2\lambda} \right \rfloor} \right \rceil = t $ since $\lambda = \frac{q}{2} = 3$. Now if there exist  $\mathbf{u}_1, \mathbf{u}_2, \mathbf{u}_3 \in \mathbb{Z}_6^k$ with the given pairwise distances and conditions, we have the distance requirement matrix 
	\[
	\mathbf{D}_f(t, \mathbf{u}_1,\mathbf{u}_2, \mathbf{u}_3) =
	\begin{bmatrix}
		0  &  2t  &  2t  \\
		2t & 0 & 2t-1  \\
		2t & 2t-1 & 0  \\
	\end{bmatrix}
	.\]
	From the generalized Plotkin bound stated in Theorem $3$, we obtain \(N_L(D, t, \mathbf{u}_1, \mathbf{u}_2, \mathbf{u}_3) \ge \frac{4}{3(3^2 - 1)} (2t + 2t + 2t - 1) \ge t - \frac{1}{6}.
	\) Since \( N_L(D, t, \mathbf{u}_1, \mathbf{u}_2, \mathbf{u}_3) \) is an integer by definition, it follows that \( N_L(D, t, \mathbf{u}_1, \mathbf{u}_2, \mathbf{u}_3) \ge t.\) The redundancy \( r \) provided by \textit{Construction~\ref{Cons4}} attains this lower bound, and hence, \textit{Construction~\ref{Cons4}} achieves the optimal redundancy of \( r_L^l(k, t) = t \) for any $t$ and $k$.
\end{proof}
Any function can be expressed as a locally-bounded function for some values of $\rho$ and $\lambda$. In \cite{GA}, the authors have shown that the Lee weight and the Lee weight distribution functions satisfy the contiguous block condition in Lemma~\ref{Cont_LB}. Since explicit constructions are given for Lee weight and Lee weight distribution functions, it is worth comparing these constructions with \textit{Construction~\ref{Cons4}} for corresponding locally-bounded functions. These comparisons are given in the following remarks.
\begin{rem}
	\label{R3}
 Lee weight function can be expressed as a locally $(2t, 4t+1)$-bounded function\cite{GA}. For this function, we can obtain a redundancy of ${\left \lceil \frac{t}{\left \lfloor \frac{q}{8t+2} \right \rfloor} \right \rceil}$ using Construction~\ref{Cons4} for any $q \ge 8t+2$. This redundancy will be less than or equal to the redundancy obtained using Construction~\ref{Cons1}. However,  Construction~\ref{Cons1} gives FCLCs for $2t+3 \le q \le 8t+2$, which is not obtainable from Construction~\ref{Cons4}. Therefore, the redundancy for the FCLCs of the Lee weight function is
		\begin{equation}
			r =
			\begin{cases}
				t, & \text{if } 2t+3 \le q \le 8t+2 \text{ and } q \text{ is odd},\\
				t+1, & \text{if } 2t+3 \le q \le 8t+2 \text{ and } q \text{ is even},\\
				\left\lceil \frac{t}{\left\lfloor \frac{q}{8t+2} \right\rfloor} \right\rceil, & \text{if } q \ge 8t+2 .
			\end{cases}
	\end{equation}
\end{rem}
\begin{rem}
	\label{R4}
  Lee weight distribution function can be expressed as a locally $(2t, \left \lfloor \frac{4t}{T} \right \rfloor +2)$-bounded function \cite{GA}. For this function, we can obtain a redundancy of ${\left \lceil \frac{t}{\left \lfloor \frac{q}{\frac{8t}{T}+4} \right \rfloor} \right \rceil}$ using Construction~\ref{Cons4} for any $q \ge 2 (\left \lfloor \frac{4t}{T} \right \rfloor +2) $. This redundancy will be less than or equal to the redundancy obtained using  Construction~\ref{Cons1}. However,   Construction~\ref{Cons1} gives FCLCs for $5\le q \le 2(\left \lfloor \frac{4t}{T} \right \rfloor +2)$, which is not obtainable from Construction~\ref{Cons4}. Therefore, the redundancy for the FCLCs of the Lee weight distribution function is
		\begin{equation}
			r =
			\begin{cases}
				t, & \text{if } \frac{2(t+1)}{T}+3 \le q \le  2(\left \lfloor \frac{4t}{T} \right \rfloor +2) \text{ and } q \text{ is odd},\\
				t+1, & \text{if } max(2t, \frac{2(t+1)}{T}+4) \le q \le  2(\left \lfloor \frac{4t}{T} \right \rfloor +2) \text{ and } q \text{ is even},\\
				\left \lceil \frac{t}{\left \lfloor \frac{q}{\frac{8t}{T}+4} \right \rfloor} \right \rceil, & \text{if } q \ge  2(\left \lfloor \frac{4t}{T} \right \rfloor +2).
			\end{cases}
			\end{equation}
\end{rem}
 While the modular sum function can be expressed as a locally $(2t, q)$-bounded function,  Construction~\ref{Cons4} is not applicable since $\lambda > \frac{q}{2}$.
\section{Redundancy Comparisons}
\label{Redundancy Comparisons}
At the beginning, we claimed that FCLCs can significantly reduce redundancy when the message length is large and the image of the target function is relatively small. In this section, we first substantiate this claim by demonstrating that for functions belonging to any of the three previously discussed classes, the redundancy is indeed lower compared to both classical Lee error-correcting codes (ECC on data) and error-correcting codes applied directly to function values (ECC on function values), as illustrated in Table~\ref{Tab_Red_Comp}. 

\begin{table*}[!htbp]
	\centering
	\footnotesize
	\resizebox{\textwidth}{!}{
	\begin{tabular}{llccc}
		\toprule
		\textbf{ Function} & \textbf{Parameters} & 
		\shortstack{\textbf{\makecell{ECC on \\ Data \\ (Lower Bound)}}} & 
		\shortstack{\textbf{\makecell{ECC on \\Function Values \\ (Lower Bound)}}} & 
		\shortstack{\textbf{ FCLCs (Exact Redundancy)}} \\
		\midrule
		\makecell{Lee weight \\ $\mathrm{w_L}(u)$ }& \makecell{$ E = k \left\lfloor \frac{q}{2} \right\rfloor + 1,$ \\ }& $\log_q V_{t}^{(n)}$ & $\log_q{[[k \left\lfloor \frac{q}{2} \right\rfloor + 1] V_{t}^{(n)}]}$ & {\scriptsize $\begin{cases}
			\begin{aligned}
				t,\hspace{1cm} & \makecell{ 2t+3 \le q \le 8t+2  \text{ and odd } q,}\\
				t+1, \hspace{1cm} & \makecell{ 2t+3 \le q \le 8t+2  \text{ and } \text{ even } q,}\\
				{\tiny \left\lceil \frac{t}{\left\lfloor \frac{q}{8t+2} \right\rfloor} \right\rceil},\hspace{0.5cm} &  q \ge 8t+2.
			\end{aligned}
		\end{cases}$ } \\
		\makecell{Lee weight \\  distribution $\Delta_T(u)$} & \makecell{$E = \frac{k \left\lfloor \frac{q}{2} \right\rfloor + 1}{T},$}
		& $\log_q V_{t}^{(n)}$ & $\log_q{[[\frac{k \left\lfloor \frac{q}{2} \right\rfloor + 1}{T}] V_{t}^{(n)}]}$ & {\scriptsize $\begin{cases}
		\begin{aligned}
			t,\hspace{1cm} & \makecell{  \frac{2(t+1)}{T}+3 \le q  \le  2(\left \lfloor \frac{4t}{T} \right \rfloor +2) \text{ and odd } q,}\\
			t+1, \hspace{1cm} & \makecell{ max(2t, \frac{2(t+1)}{T}+4) \le q  \le  2(\left \lfloor \frac{4t}{T} \right \rfloor +2) \\  \text{ and even } q,}\\
			{\tiny{\left \lceil \frac{t}{\left \lfloor \frac{q}{\frac{8t}{T}+4} \right \rfloor} \right \rceil},}\hspace{0.5cm} &  q \ge 2(\left \lfloor \frac{4t}{T} \right \rfloor +2).
		\end{aligned}
		\end{cases}$ }\\
		\makecell{Modular sum \\ $S_{m}(\boldsymbol{u})$} & \makecell{$ E = q,$ \\
		$q \ge 5$ and $t \le \frac{q-3}{2}.$}
		& $\log_q V_{t}^{(n)}$ & $\log_q{[q V_{t}^{(n)}]}$ & $\begin{cases}
			\begin{aligned}
				&t,  \text{ if } q \text{ is odd}, \\
				&t+1,  \text{ if } q \text{ is even}.
			\end{aligned}
		\end{cases}$ \\
		\multicolumn{5}{l}{\makecell{Where $V_{1}^{(n)} = 1+2n$ for any $q \ge 3$, $V_{2}^{(n)} = 1+2n+2n^2$ for any $q \ge 5$, $V_{t}^{(n)} = \sum_{i=0}^{\min(n,t)} \binom{n}{i} 2^i \binom{t}{i} \text{ for } t \le \frac{q-1}{2}$, \\ and $n=k+r$ for ECC on data and $n=r$ for ECC on function values.}}\\
		\midrule
	\end{tabular}}
	\caption{Redundancy comparison for different functions over $\mathbb{Z}_q$.}
	\label{Tab_Red_Comp}
\end{table*}

First, we evaluate the redundancy of systematic classical Lee error-correcting codes labeled as the column
``ECC on Data” in Table~\ref{Tab_Red_Comp}, given by $ r_{ECC} =n - k$, where $n, k$ denote the codeword length and information vector length, respectively. Using the sphere-packing bound\cite{CW}, we have ${ q^{n} \ge q^{k}V_{t}^{(n)}}$, which leads to the lower bound, $ r_{ECC} =n - k \ge \log_q V_{t}^{(n)}$,  where $V_{t}^{(n)}$ denotes the volume of a Lee ball of radius $t$. $V_{1}^{(n)} = 1+2n$ for any $q \ge 3$ and $V_{2}^{(n)} = 1+2n+2n^2$ for any $q \ge 5$ \cite{CW}. For $t \le \frac{q-1}{2}$, the volume can be approximated as  $V_{t}^{(n)} = \sum_{i=0}^{min(n,t)} \binom{n}{i} 2^i \binom{t}{i}$ \cite{RR}. Therefore, the redundancy satisfies $r_{ECC} =n - k \ge \log_q [\sum_{i=0}^{min(n,t)} \binom{n}{i} 2^i \binom{t}{i}]$ for classical ECC on data. Importantly, this redundancy bound is independent of the specific function being corrected and applies uniformly across all functions.

Next, we consider a direct approach of encoding only the function values, referred to as ``ECC on function values" in Table~\ref{Tab_Red_Comp}. To determine the redundancy in this case, we analyze the minimum codelength $n$ required for a Lee code with a specified number of codewords $E$ (equal to the number of distinct function values) and minimum distance $2t+1$. The resulting codeword $\boldsymbol{c}$ is then appended to $\boldsymbol{u}$, to ensure systematic encoding, forming the transmitted codeword ($\boldsymbol{u}$, $\boldsymbol{c}$). Applying the sphere-packing bound\cite{CW} again, we obtain ${ q^{n} \ge EV_{t}^{(n)}}$, which implies \( n \ge \log_q{[E V_{t}^{(n)}]}\). This results in a lower bound on the redundancy for ECC on function values, i.e., \( r_{f} \ge \log_q{[E V_{t}^{(n)}]}\), which depends explicitly on the function through the parameter $E$.
The following example illustrates the minimal redundancy required for FCLCs compared to the lower bound on redundancy for ECC on data and ECC on function values. 
\begin{example}
 Consider the FCLC for the Lee weight function for $q=5, k=2$ and $t=1$. The required redundancy in this case is $t=1$, as noted from Table~\ref{Tab_Red_Comp}. For the same function with the same parameters, the lower bound on redundancy for ECC on data is $2$ and that of ECC on function values is $3$. 
\end{example}
To have clear illustrations of the redundancy advantage of FCLCs, we next generate plots showing the redundancy required $r$  versus the number of errors corrected $t$ for the Lee weight (Fig.~\ref{Red_LW}), Lee weight distribution (Fig.~\ref{Red_LWD}), and modular sum functions (Fig.~\ref{Red_MS}). These plots are shown below separately for odd and even values of $q$.
\begin{figure}[htbp]
	\centering
	\begin{subfigure}{0.48\textwidth}
		\centering
		\includegraphics[width=\linewidth]{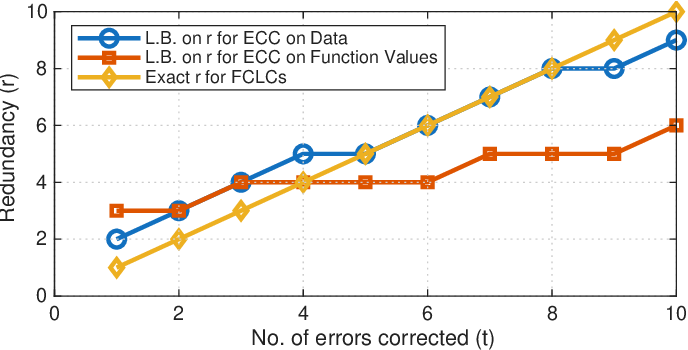}
		\caption{$q=25$, $k=25$, $t=10$}
	\end{subfigure}
	\hfill
	\begin{subfigure}{0.48\textwidth}
		\centering
		\includegraphics[width=\linewidth]{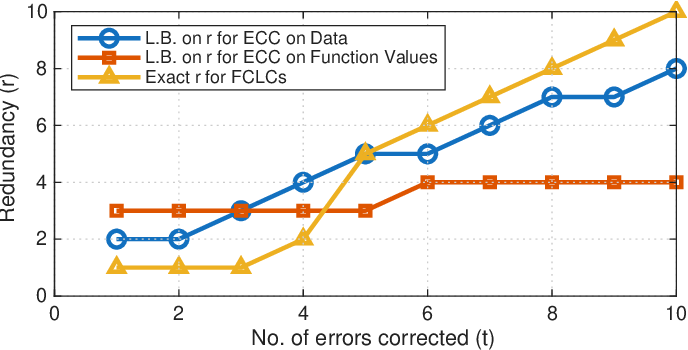}
		\caption{$q=82$, $k=50$, $t=10$}
	\end{subfigure}
	
	\caption{Redundancy vs. No. of errors corrected for Lee weight function}
	\label{Red_LW}
\end{figure}

\begin{figure}[htbp]
	\centering
	\begin{subfigure}{0.48\textwidth}
		\centering
		\includegraphics[width=\linewidth]{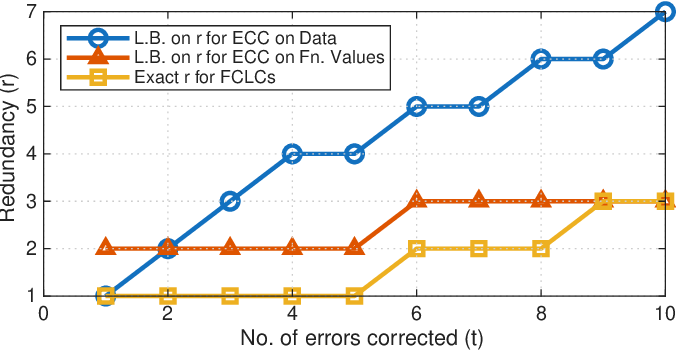}
		\caption{$q=35$, $k=15$, $t=10$, $T=18$}
	\end{subfigure}
	\hfill
	\begin{subfigure}{0.48\textwidth}
		\centering
		\includegraphics[width=\linewidth]{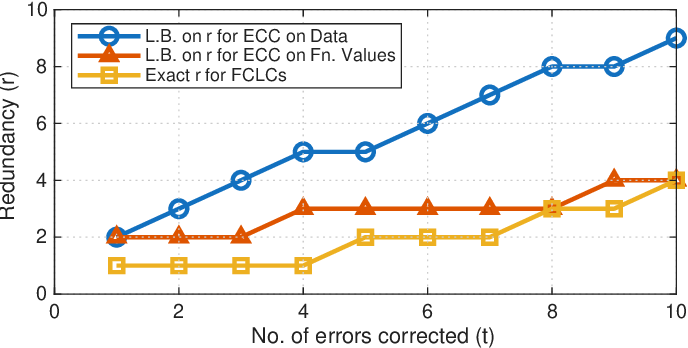}
		\caption{$q=30$, $k=30$, $t=10$, $T=16$}
	\end{subfigure}
	
	\caption{Redundancy vs. No. of errors corrected for Lee weight distribution function}
	\label{Red_LWD}
\end{figure}

\begin{figure}[htbp]
	\centering
	\begin{subfigure}{0.48\textwidth}
		\centering
		\includegraphics[width=\linewidth]{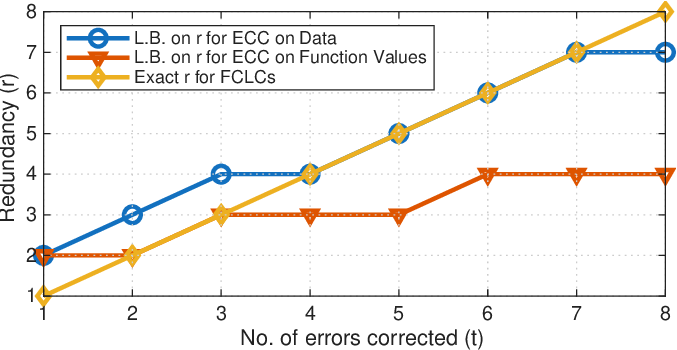}
		\caption{$q=19$, $k=15$, $t=8$}
	\end{subfigure}
	\hfill
	\begin{subfigure}{0.48\textwidth}
		\centering
		\includegraphics[width=\linewidth]{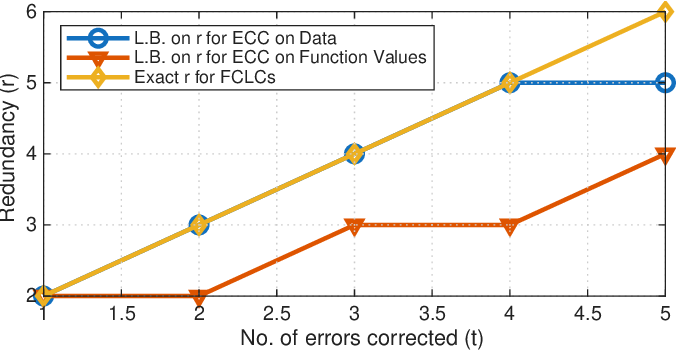}
		\caption{$q=14$, $k=12$, $t=5$}
	\end{subfigure}
	
	\caption{Redundancy vs. No. of errors corrected for modular sum function}
	\label{Red_MS}
\end{figure}
From the figures, we observe that for smaller values of $t$, the exact redundancy values achieved by our FCLC constructions are smaller than both the lower bound on the redundancy for ECC on data and the lower bound on the redundancy for ECC on function values. This clearly illustrates the redundancy advantage of our FCLC constructions in this regime. When $t$ is large, the lower bounds on the redundancy for ECC on function values and ECC on data become equal to or smaller than the exact redundancy achieved by the FCLC constructions. It should be noted that the redundancy for ECC on data and the redundancy for ECC on function values are only lower bounds. In general, there are no known Lee codes that achieve the lower bound on the redundancy for ECC on data, and the known constructions typically require larger redundancies. This is illustrated with the following examples.

In \cite{NCB}, Berlekamp provides constructions of negacyclic codes for the Lee metric. As given in \cite{NCB}, for $q=17$ and $t=6$, there exists a negacyclic Lee code with parameters $k=14$ and redundancy $r=10$. For the same parameters $q=17$, $k=14$, and $t=6$, the lower bound on redundancy for ECC on data is $6$. However, the Berlekamp construction requires redundancy $10$, whereas the redundancy obtained from the FCLC construction is $6$. Similarly, for $q=17$ and $t=7$,  there exists a negacyclic Lee code with parameters $k=12$ and $r=12$. For $q=17$, $k=12$ and $t=7$, although the lower bound on redundancy for ECC on data is $6$, redundancy given by Berlekamp construction is $12$,   whereas the redundancy obtained from the FCLC construction is only $7$. 

 Next, we compare the upper bounds on optimal redundancy for FCLCs of Lee weight and Lee weight distribution functions. Since the optimal redundancy is defined as the minimum redundancy over all feasible FCLCs, the redundancy achieved by an explicit construction directly provides an upper bound on the optimal redundancy for the considered function. Therefore, by Remark \ref{R3}, we obtain an  upper bound on the optimal redundancy of FCLCs for the Lee weight function as follows.  

\begin{equation*}
	r_L^\mathrm{w}(q, k,t) \le \begin{cases}
		t, & \text{if } 2t+3 \le q \le 8t+2 \text{ and } q \text{ is odd},\\
		t+1, & \text{if } 2t+3 \le q \le 8t+2 \text{ and } q \text{ is even},\\
		\left\lceil \frac{t}{\left\lfloor \frac{q}{8t+2} \right\rfloor} \right\rceil, & \text{if } q \ge 8t+2.
	\end{cases}
\end{equation*} 
An upper bound on the optimal redundancy of FCLC for the Lee weight function is given in \cite{GA}, as
$r_L^\mathrm{w}(q, k,t) \le (4t+1) \left\lceil \frac{t}{\lfloor \frac{q}{2} \rfloor} \right\rceil $. For all $q \ge 2$ and $t \ge 1$, we have $\left\lceil \frac{t}{\lfloor \frac{q}{2} \rfloor} \right\rceil \ge 1$. Therefore, $ (4t+1) \left\lceil \frac{t}{\lfloor \frac{q}{2} \rfloor} \right\rceil \ge (4t+1) > t+1$. That is, the upper bound on the optimal redundancy of FCLCs for the Lee weight function obtained from our explicit construction is tighter than the upper bound given in \cite{GA}. 

By Remark~\ref{R4}, we obtain an  upper bound on the optimal redundancy of FCLCs for the Lee weight distribution function as follows.

\begin{equation*}
	r_L^\Delta(q, k,t) \le \begin{cases}
		t, & \text{if } \frac{2(t+1)}{T}+3 \le q \le  2(\left \lfloor \frac{4t}{T} \right \rfloor +2) \text{ and } q \text{ is odd},\\
		t+1, & \text{if } max(2t, \frac{2(t+1)}{T}+4) \le q \le  2(\left \lfloor \frac{4t}{T} \right \rfloor +2) \text{ and } q \text{ is even},\\
		\left \lceil \frac{t}{\left \lfloor \frac{q}{\frac{8t}{T}+4} \right \rfloor} \right \rceil, & \text{if } q \ge  2(\left \lfloor \frac{4t}{T} \right \rfloor +2).
	\end{cases}
\end{equation*} 

An upper bound on the optimal redundancy of FCLC for the Lee weight distribution function for $T < k \lfloor \frac{q}{2} \rfloor$ is given in \cite{GA}, as $ 	r_L^\Delta(q, k,t) \le \left( \left\lfloor \frac{4t}{T} \right\rfloor +2 \right) \left\lceil \frac{t}{\lfloor \frac{q}{2} \rfloor} \right\rceil$. This upper bound and the upper bound obtained by Lemma \ref{Lem_Cons1_LWD} are simultaneously valid when $t \le \frac{(q-3)T}{2}-1$ and $T < k \lfloor \frac{q}{2} \rfloor$, which holds for a large range of parameters. Since  $\left\lceil \frac{t}{\lfloor \frac{q}{2} \rfloor} \right\rceil \ge  \frac{t}{\lfloor \frac{q}{2} \rfloor} $ and $\left\lfloor \frac{4t}{T} \right\rfloor +2  \ge  \frac{4t}{T}  +1 $, we have $\left( \left\lfloor \frac{4t}{T} \right\rfloor +2 \right) \left\lceil \frac{t}{\lfloor \frac{q}{2} \rfloor} \right\rceil \ge \left(\frac{4t}{T}  +1\right) \frac{t}{\lfloor \frac{q}{2} \rfloor}$. And, we have $t \le \left(\frac{4t}{T}  +1\right) \frac{t}{\lfloor \frac{q}{2} \rfloor}$ when $t  \ge \frac{T}{4} \left(\lfloor \frac{q}{2} \rfloor -1\right)$. Therefore, in the range $\frac{T}{4} \left(\lfloor \frac{q}{2} \rfloor -1\right) \le t \le \frac{(q-3)T}{2}-1$, the upper bound on the optimal redundancy of FCLCs for the Lee weight distribution function obtained from our explicit construction is tighter than the upper bound given in \cite{GA}.
\section{Conclusion}
\label{Conclusion}
This work presented explicit constructions of FCLCs for the Lee weight function, the Lee weight distribution function, the modular sum function and the locally-bounded function. We demonstrated that the proposed constructions achieve optimal redundancy for certain parameters. We also proposed a Plotkin-like bound for irregular Lee-distance codes. A comparative analysis with classical Lee error-correcting codes (ECC on data) and codes that correct errors directly in function values (ECC on function values) showed that FCLCs offer significant reductions in redundancy while preserving function correctness under error-prone conditions. Further research directions include the study of FCLCs for new classes of functions in the Lee metric as well as the investigation of tighter lower and upper bounds on the optimal redundancy.

\section*{Acknowledgement}
This work was supported partly by the Science and Engineering Research Board (SERB) of Department of Science and Technology (DST), Government of India, through J.C Bose National Fellowship to B. Sundar Rajan.

\begin{appendices}
	\section{Proof of Lemma \ref{Lem_Cons1_FCLC}}\label{appendix:Lem_Cons1_FCLC}

	To prove that \textit{Construction~\ref{Cons1}}  results in an FCLC for the Lee weight function, we have to show that for all \( \mathbf{u}, \mathbf{u'} \in \mathbb{Z}_q^k \) with \( f(\mathbf{u}) \neq f(\mathbf{u'}) \), \textit{Construction~\ref{Cons1}} ensures \(d_L(\mathrm{Enc}(\mathbf{u}), \mathrm{Enc}(\mathbf{u'})) \geq 2t + 1 \). Since \( f(\mathbf{u}) \neq f(\mathbf{u'}) \), WLOG assume \( f(\mathbf{u}) > f(\mathbf{u'}) \).  The proof is done for odd and even $q$ separately.  We consider the following cases. \\
	
	\noindent $1.$ Proof for odd values of $q$. \\
	\noindent $\bullet$ Case $1$: $d_{L}({p}_{\boldsymbol{u}}, {p}_{\boldsymbol{u'}}) =0$.  By $(\ref{Eq_Cons1})$, $d_{L}({p}_{\boldsymbol{u}}, {p}_{\boldsymbol{u'}}) =0$ implies \((2f(\boldsymbol{u})) \bmod q = (2f(\boldsymbol{u'})) \bmod q\). That is, $f(\mathbf{u}) \equiv f(\mathbf{u'}) \pmod q$. Thus, $f(\mathbf{u}) - f(\mathbf{u'})  \in \{ 0,  q,  2q, \ldots \}$. Since $f(\mathbf{u}) \neq f(\mathbf{u'})$ and \( f(\mathbf{u}) > f(\mathbf{u'}) \), we have  $f(\mathbf{u}) - f(\mathbf{u'}) \ge q $. Lee distance satisfies triangle inequality \cite{RR}. Therefore, $d_L(\mathbf{u},0) \le d_L(\mathbf{u}, \mathbf{u'}) + d_L(\mathbf{u'},0)$, which implies $d_L(\mathbf{u}, \mathbf{u'})  \ge d_L(\mathbf{u},0) - d_L(\mathbf{u'},0) = \mathrm{w_L}(\boldsymbol{u}) - \mathrm{w_L}(\boldsymbol{u'})=f(\mathbf{u})-f(\mathbf{u'})$.  Therefore, $f(\mathbf{u}) - f(\mathbf{u'}) \ge q$ implies $d_{L}(\boldsymbol{u}, \boldsymbol{u'}) \ge q $. By  the condition in Lemma \ref{Lem_Cons1_FCLC}, $q \ge 2t+3$.  Therefore, $d_{L}(\mathrm{Enc}(\boldsymbol{u}), \mathrm{Enc}(\boldsymbol{u'})) = d_{L}(\boldsymbol{u}, \boldsymbol{u'}) + d_{L}(\{{p}_{\boldsymbol{u}}\}^{t}, \{{p}_{\boldsymbol{u'}}\}^{t}) \ge (2t+3) + 0 > 2t+1$. 
	
	\noindent $\bullet$ Case $2$: $d_{L}({p}_{\boldsymbol{u}}, {p}_{\boldsymbol{u'}}) =1$. By definition of the Lee distance, $d_{L}({p}_{\boldsymbol{u}}, {p}_{\boldsymbol{u'}}) =1$ implies $min(|{p}_{\boldsymbol{u}} - {p}_{\boldsymbol{u'}}|, q - |{p}_{\boldsymbol{u}} - {p}_{\boldsymbol{u'}}|) = 1$. That is, $|{p}_{\boldsymbol{u}} - {p}_{\boldsymbol{u'}}|$ is either $1$ or $q-1$. Therefore, for odd $q$ by $(\ref{Eq_Cons1})$ we have, $|(2f(\boldsymbol{u})) \bmod q - (2f(\boldsymbol{u'})) \bmod q|$ is either $1$ or $q-1$.  Thus, $(2f(\boldsymbol{u})) \bmod q - (2f(\boldsymbol{u'})) \bmod q \in \{\pm 1, \pm (q-1)\}$. This implies, $2 (f(\boldsymbol{u}) -  f(\boldsymbol{u'}) ) \equiv \pm 1 \pmod q$. That is $f(\mathbf{u}) - f(\mathbf{u'}) \equiv \pm 2^{-1} \pmod q $, where $2^{-1}$ denotes the multiplicative inverse of $2$ in $\mathbb{Z}_q$. Since $q$ is odd, $2^{-1}$  is $\frac{q+1}{2} \pmod q$ and $-(\frac{q+1}{2})=\frac{q-1}{2} \pmod q$. Therefore, we have $f(\mathbf{u}) - f(\mathbf{u'}) \ge \frac{q-1}{2} $. Thus, $ d_{L}(\boldsymbol{u}, \boldsymbol{u'}) \ge f(\mathbf{u}) - f(\mathbf{u'}) \ge \frac{q-1}{2}$. Therefore, $d_{L}(\mathrm{Enc}(\boldsymbol{u}), \mathrm{Enc}(\boldsymbol{u'})) = d_{L}(\boldsymbol{u}, \boldsymbol{u'}) + d_{L}(\{{p}_{\boldsymbol{u}}\}^{t}, \{{p}_{\boldsymbol{u'}}\}^{t}) \ge \frac{q-1}{2} + t \ge \frac{2t+3-1}{2} +t = 2t+1$. 
	
	\noindent $\bullet$ Case $3$:  $d_{L}({p}_{\boldsymbol{u}}, {p}_{\boldsymbol{u'}}) \ge 2$. For any $\mathbf{u},\mathbf{u'}$ with \( f(\mathbf{u}) \neq f(\mathbf{u'}) \), \(d_{L}(\boldsymbol{u}, \boldsymbol{u'}) \ge 1 \). Therefore, 
	$d_{L}(\mathrm{Enc}(\boldsymbol{u}), \mathrm{Enc}(\boldsymbol{u'})) = d_{L}(\boldsymbol{u}, \boldsymbol{u'}) + d_{L}(\{{p}_{\boldsymbol{u}}\}^{t}, \{{p}_{\boldsymbol{u'}}\}^{t}) \ge 1 + 2t = 2t+1$. \\
	
	\noindent $2.$ Proof for even values of $q$. \\
	\noindent $\bullet$ Case $1$: $d_{L}({p}_{\boldsymbol{u}}, {p}_{\boldsymbol{u'}}) =0$.  In this case, we use the fact that, if $a \equiv b \pmod q$ and $q$ is even, then $a$ and $b$ are either both even or both odd. By $(\ref{Eq_Cons1})$, $d_{L}({p}_{\boldsymbol{u}}, {p}_{\boldsymbol{u'}}) =0$ implies either \((2f(\boldsymbol{u})) \bmod q = (2f(\boldsymbol{u'})) \bmod q\), or \((2f(\boldsymbol{u})+1) \bmod q = (2f(\boldsymbol{u'})+1) \bmod q\), or \((2f(\boldsymbol{u})) \bmod q = (2f(\boldsymbol{u'})+1) \bmod q\). The case \((2f(\boldsymbol{u})) \bmod q = (2f(\boldsymbol{u'})+1) \bmod q\) is not possible since left side is even and right side is odd. Therefore, either \((2f(\boldsymbol{u})) \bmod q = (2f(\boldsymbol{u'})) \bmod q\), or \((2f(\boldsymbol{u})+1) \bmod q = (2f(\boldsymbol{u'})+1) \bmod q\). That is, in both cases $f(\mathbf{u}) \equiv f(\mathbf{u'}) \pmod {q/2}$. Therefore, $f(\mathbf{u}) - f(\mathbf{u'}) \in \{ 0, \frac{q}{2},  q, \ldots\}$. Since $f(\mathbf{u}) \neq f(\mathbf{u'})$, $f(\mathbf{u}) - f(\mathbf{u'}) \neq 0$. Suppose $f(\mathbf{u}) = f(\mathbf{u'}) + \frac{q}{2}$. If $f(\mathbf{u'}) \bmod q \in \{0,1,...,(\frac{q}{2}-1) \}$, then $f(\mathbf{u}) \bmod q \in \{\frac{q}{2},(\frac{q}{2}+1),...,(q-1)\}$. In this case, by $(\ref{Eq_Cons1})$, \((2f(\boldsymbol{u})) \bmod q = (2f(\boldsymbol{u'})+1) \bmod q\), which is not possible when $d_{L}({p}_{\boldsymbol{u}}, {p}_{\boldsymbol{u'}}) =0$. Therefore, we have $f(\mathbf{u}) - f(\mathbf{u'}) \ge q$. This implies $d_{L}(\boldsymbol{u}, \boldsymbol{u'}) \ge q $. By  the condition in Lemma \ref{Lem_Cons1_FCLC}, $q \ge 2t+3$.  Therefore, $d_{L}(\mathrm{Enc}(\boldsymbol{u}), \mathrm{Enc}(\boldsymbol{u'})) = d_{L}(\boldsymbol{u}, \boldsymbol{u'}) + d_{L}(\{{p}_{\boldsymbol{u}}\}^{t}, \{{p}_{\boldsymbol{u'}}\}^{t}) + d_{L}(\{{p'}_{\boldsymbol{u}}\}, \{{p'}_{\boldsymbol{u'}}\}) \ge d_{L}(\boldsymbol{u}, \boldsymbol{u'}) + d_{L}(\{{p}_{\boldsymbol{u}}\}^{t}, \{{p}_{\boldsymbol{u'}}\}^{t}) \ge (2t+3) + 0 > 2t+1$. 
	
	\noindent $\bullet$ Case $2$: $d_{L}({p}_{\boldsymbol{u}}, {p}_{\boldsymbol{u'}}) =1$. From  definition of the Lee distance, we have $d_{L}({p}_{\boldsymbol{u}}, {p}_{\boldsymbol{u'}}) =1$ implies $|{p}_{\boldsymbol{u}} - {p}_{\boldsymbol{u'}}|$ is either $1$ or $q-1$. By $(\ref{Eq_Cons1})$, the condition $d_{L}({p}_{\boldsymbol{u}}, {p}_{\boldsymbol{u'}}) =1$ is possible only if one of $f(\boldsymbol{u}) \bmod q$ and $f(\boldsymbol{u'}) \bmod q$ lies in $\{0,1,...,(\frac{q}{2}-1)\}$, while the other lies in $\{\frac{q}{2},(\frac{q}{2}+1),...,(q-1)\}$. \\ Subcase 1: $f(\boldsymbol{u}) \bmod q \in \{0,1,...,(\frac{q}{2}-1)\}$ and $f(\boldsymbol{u'}) \bmod q \in \{\frac{q}{2},(\frac{q}{2}+1),...,(q-1)\}$. In this case by $(\ref{Eq_Cons1})$ we have, $|(2f(\boldsymbol{u})) \bmod q - (2f(\boldsymbol{u'})+1) \bmod q|$ is either $1$ or $q-1$. Thus, $(2f(\boldsymbol{u})) \bmod q - (2f(\boldsymbol{u'})) \bmod q -1 \in \{\pm 1, \pm (q-1)\}$. This implies, $2 (f(\boldsymbol{u}) -  f(\boldsymbol{u'}) ) \bmod q \in  \{0,2,q,2-q\}$. That is either $f(\mathbf{u}) - f(\mathbf{u'}) \equiv 0 \pmod {\frac{q}{2}} $ or $f(\mathbf{u}) - f(\mathbf{u'}) \equiv 1 \pmod {\frac{q}{2}} $. Equivalently, either $f(\boldsymbol{u}) -f(\boldsymbol{u'}) \in \{ 0, \frac{q}{2}, q, \frac{3q}{2}, \ldots \}$ or $f(\boldsymbol{u}) -f(\boldsymbol{u'}) \in \{ 1, 1+\frac{q}{2}, 1+q, 1+\frac{3q}{2}, \ldots \}$. The case of $f(\boldsymbol{u}) -f(\boldsymbol{u'}) = 1 $ is exceptional and will be addressed separately. Excluding this case, and since  $f(\mathbf{u}) \ne f(\mathbf{u'})$, we have $f(\mathbf{u}) - f(\mathbf{u'}) \ge \frac{q}{2}$. Thus, $ d_{L}(\boldsymbol{u}, \boldsymbol{u'}) \ge f(\mathbf{u}) - f(\mathbf{u'}) \ge \frac{q}{2}$. \\  
	Subcase 2: $f(\boldsymbol{u}) \bmod q \in \{\frac{q}{2},(\frac{q}{2}+1),...,(q-1)\} $ and $f(\boldsymbol{u'}) \bmod q \in \{0,1,...,(\frac{q}{2}-1)\}$. In this case by $(\ref{Eq_Cons1})$ we have, $|(2f(\boldsymbol{u})+1) \bmod q - (2f(\boldsymbol{u'})) \bmod q|$ is either $1$ or $q-1$. From this we obtain, either $f(\mathbf{u}) - f(\mathbf{u'}) \equiv 0 \pmod {\frac{q}{2}} $ or $f(\mathbf{u}) - f(\mathbf{u'}) \equiv \frac{q}{2}-1 \pmod {\frac{q}{2}} $. Therefore, $f(\mathbf{u}) - f(\mathbf{u'}) \ge \frac{q}{2}-1$. Thus, $ d_{L}(\boldsymbol{u}, \boldsymbol{u'}) \ge f(\mathbf{u}) - f(\mathbf{u'}) \ge \frac{q}{2}-1$.  
	
	Thus, when $d_{L}({p}_{\boldsymbol{u}}, {p}_{\boldsymbol{u'}}) =1$, we have $ d_{L}(\boldsymbol{u}, \boldsymbol{u'}) \ge \frac{q}{2}-1$. Therefore, $d_{L}(\mathrm{Enc}(\boldsymbol{u}), \mathrm{Enc}(\boldsymbol{u'})) = d_{L}(\boldsymbol{u}, \boldsymbol{u'}) + d_{L}(\{{p}_{\boldsymbol{u}}\}^{t}, \{{p}_{\boldsymbol{u'}}\}^{t}) + d_{L}(\{{p'}_{\boldsymbol{u}}\}, \{{p'}_{\boldsymbol{u'}}\}) \ge d_{L}(\boldsymbol{u}, \boldsymbol{u'}) + d_{L}(\{{p}_{\boldsymbol{u}}\}^{t}, \{{p}_{\boldsymbol{u'}}\}^{t}) \ge \frac{q}{2}-1 + t \ge \frac{2t+3}{2}-1+t = 2t+\frac{1}{2}$. Thus, $d_{L}(\mathrm{Enc}(\boldsymbol{u}), \mathrm{Enc}(\boldsymbol{u'})) \ge 2t+1$. 
	
	Next, consider the case when  $f(\boldsymbol{u}) -f(\boldsymbol{u'}) = 1 $. This is possible only when $f(\boldsymbol{u}) \bmod q = 0$ and $f(\boldsymbol{u'}) \bmod q = q-1$. Now, consider the encoding given in Construction \ref{Cons1}. We have, $d_{L}(\mathrm{Enc}(\boldsymbol{u}), \mathrm{Enc}(\boldsymbol{u'})) = d_{L}(\boldsymbol{u}, \boldsymbol{u'}) + d_{L}(\{{p}_{\boldsymbol{u}}\}^{t}, \{{p}_{\boldsymbol{u'}}\}^{t}) + d_{L}(\{{p'}_{\boldsymbol{u}}\}, \{{p'}_{\boldsymbol{u'}}\}) \ge 1 + t + d_{L}(0,\frac{q}{2}) = 1 + t + \frac{q}{2} \ge 1+t + t + \frac{3}{2} > 2t+1$. 
	
	\noindent $\bullet$ Case $3$:  $d_{L}({p}_{\boldsymbol{u}}, {p}_{\boldsymbol{u'}}) \ge 2$. For any $\mathbf{u},\mathbf{u'}$ with \( f(\mathbf{u}) \neq f(\mathbf{u'}) \), \(d_{L}(\boldsymbol{u}, \boldsymbol{u'}) \ge 1 \). Therefore, 
	$d_{L}(\mathrm{Enc}(\boldsymbol{u}), \mathrm{Enc}(\boldsymbol{u'})) = d_{L}(\boldsymbol{u}, \boldsymbol{u'}) + d_{L}(\{{p}_{\boldsymbol{u}}\}^{t}, \{{p}_{\boldsymbol{u'}}\}^{t}) + d_{L}(\{{p'}_{\boldsymbol{u}}\}, \{{p'}_{f\boldsymbol{u'}}\}) \ge d_{L}(\boldsymbol{u}, \boldsymbol{u'}) + d_{L}(\{{p}_{\boldsymbol{u}}\}^{t}, \{{p}_{\boldsymbol{u'}}\}^{t}) \ge 1 + 2t = 2t+1$. \hfill $\blacksquare$.
	
	\section{Proof of Lemma \ref{Lem_Cons1_LWD}}\label{appendix:Lem_Cons1_LWD}

		The proof structure follows similar to the proof of Lemma \ref{Lem_Cons1_FCLC}. Therefore, we omit the detailed proof and highlight only the differences arising due to the new function. \\
		
		\noindent $1.$ Proof for odd values of $q$. \\
		\noindent $\bullet$ Case $1$: $d_{L}({p}_{\boldsymbol{u}}, {p}_{\boldsymbol{u'}}) =0$.  As shown in the proof of Lemma \ref{Lem_Cons1_FCLC}, $d_{L}({p}_{\boldsymbol{u}}, {p}_{\boldsymbol{u'}}) =0$ implies  $f(\mathbf{u}) - f(\mathbf{u'}) \ge q $. Therefore, we have $ \frac{\mathrm{w_L}(\boldsymbol{u}) -\mathrm{w_L}(\boldsymbol{u'})}{T}   \geq \left \lfloor \frac{\mathrm{w_L}(\boldsymbol{u})}{T} \right \rfloor - \left \lfloor \frac{\mathrm{w_L}(\boldsymbol{u'})}{T} \right \rfloor -1 =f(\mathbf{u}) - f(\mathbf{u'})  -1  \geq q-1$. Since,  $ d_{L}(\boldsymbol{u}, \boldsymbol{u'}) \ge \mathrm{w_L}(\boldsymbol{u}) -\mathrm{w_L}(\boldsymbol{u'})$, we have $ d_{L}(\boldsymbol{u}, \boldsymbol{u'}) \ge (q-1) T$. Since $ (q-3) T \ge 2(t+1)$ by condition for odd $q$ in Lemma \ref{Lem_Cons1_LWD}, we have $ d_{L}(\boldsymbol{u}, \boldsymbol{u'}) \ge (q-1) T > (q-3) T \ge 2(t+1) > 2t+1$. Therefore, $d_{L}(\mathrm{Enc}(\boldsymbol{u}), \mathrm{Enc}(\boldsymbol{u'})) =  d_{L}(\boldsymbol{u}, \boldsymbol{u'}) + d_{L}(\{{p}_{\boldsymbol{u}}\}^{t}, \{{p}_{\boldsymbol{u'}}\}^{t}) > (2t+1) + 0 = 2t+1$.  
		
		\noindent $\bullet$ Case $2$: $d_{L}({p}_{\boldsymbol{u}}, {p}_{\boldsymbol{u'}}) =1$. As shown in the proof of Lemma \ref{Lem_Cons1_FCLC}, $d_{L}({p}_{\boldsymbol{u}}, {p}_{\boldsymbol{u'}}) =1$ implies $f(\mathbf{u}) - f(\mathbf{u'}) \ge \frac{q-1}{2} $. Therefore, $ d_{L}(\boldsymbol{u}, \boldsymbol{u'}) \ge \mathrm{w_L}(\boldsymbol{u}) -\mathrm{w_L}(\boldsymbol{u'}) \ge (f(\mathbf{u}) - f(\mathbf{u'})  -1) T \ge \left(\frac{q-1}{2}-1\right)T= \left(\frac{q-3}{2}\right)T$. Since $ (q-3) T \ge 2(t+1)$, we have  $ d_{L}(\boldsymbol{u}, \boldsymbol{u'}) \ge \left(\frac{q-3}{2}\right)T \ge t+1$. Therefore, $d_{L}(\mathrm{Enc}(\boldsymbol{u}), \mathrm{Enc}(\boldsymbol{u'})) = d_{L}(\boldsymbol{u}, \boldsymbol{u'}) + d_{L}(\{{p}_{\boldsymbol{u}}\}^{t}, \{{p}_{\boldsymbol{u'}}\}^{t}) \ge (t+1) + t = 2t+1$. 
		
		\noindent $\bullet$ Case $3$:  $d_{L}({p}_{\boldsymbol{u}}, {p}_{\boldsymbol{u'}}) \ge 2$. The argument is same as that of Case $3$ for odd values of $q$ in the proof of Lemma \ref{Lem_Cons1_FCLC}. \\
		
		\noindent $2.$ Proof for even values of $q$. \\
		\noindent $\bullet$ Case $1$: $d_{L}({p}_{\boldsymbol{u}}, {p}_{\boldsymbol{u'}}) =0$.  As shown in the proof of Lemma \ref{Lem_Cons1_FCLC}, for even $q$, $d_{L}({p}_{\boldsymbol{u}}, {p}_{\boldsymbol{u'}}) =0$ implies  $f(\mathbf{u}) - f(\mathbf{u'}) \ge q $. Therefore, as shown in Case $1$ for odd $q$, we have $ d_{L}(\boldsymbol{u}, \boldsymbol{u'}) \ge (q-1) T$. Since $ (q-4) T \ge 2(t+1)$ by condition for even $q$ in Lemma \ref{Lem_Cons1_LWD}, we have $ d_{L}(\boldsymbol{u}, \boldsymbol{u'}) \ge (q-1) T > (q-4) T \ge 2(t+1) > 2t+1$.  Therefore, $d_{L}(\mathrm{Enc}(\boldsymbol{u}), \mathrm{Enc}(\boldsymbol{u'})) = d_{L}(\boldsymbol{u}, \boldsymbol{u'}) + d_{L}(\{{p}_{\boldsymbol{u}}\}^{t}, \{{p}_{\boldsymbol{u'}}\}^{t}) + d_{L}(\{{p'}_{\boldsymbol{u}}\}, \{{p'}_{\boldsymbol{u'}}\}) \ge d_{L}(\boldsymbol{u}, \boldsymbol{u'}) + d_{L}(\{{p}_{\boldsymbol{u}}\}^{t}, \{{p}_{\boldsymbol{u'}}\}^{t}) \ge 4(t+1) + 0 > 2t+1$. 
		
		\noindent $\bullet$ Case $2$: $d_{L}({p}_{\boldsymbol{u}}, {p}_{\boldsymbol{u'}}) =1$. As shown in the proof of Lemma \ref{Lem_Cons1_FCLC}, for even $q$, $d_{L}({p}_{\boldsymbol{u}}, {p}_{\boldsymbol{u'}}) =1$ implies $f(\mathbf{u}) - f(\mathbf{u'}) \ge \frac{q}{2}-1 $, in all cases except one special case mentioned in the proof of Lemma \ref{Lem_Cons1_FCLC}. Therefore, $ d_{L}(\boldsymbol{u}, \boldsymbol{u'}) \ge \mathrm{w_L}(\boldsymbol{u}) -\mathrm{w_L}(\boldsymbol{u'}) \ge (f(\mathbf{u}) - f(\mathbf{u'})  -1) T \ge \left(\frac{q}{2}-1-1\right)T=\left(\frac{q-4}{2}\right)T$. Since $ (q-4) T \ge 2(t+1)$ we have  $ d_{L}(\boldsymbol{u}, \boldsymbol{u'}) \ge \left(\frac{q-4}{2}\right)T \ge t+1$. Therefore, $d_{L}(\mathrm{Enc}(\boldsymbol{u}), \mathrm{Enc}(\boldsymbol{u'})) = d_{L}(\boldsymbol{u}, \boldsymbol{u'}) + d_{L}(\{{p}_{\boldsymbol{u}}\}^{t}, \{{p}_{\boldsymbol{u'}}\}^{t}) + d_{L}(\{{p'}_{\boldsymbol{u}}\}, \{{p'}_{\boldsymbol{u'}}\}) \ge d_{L}(\boldsymbol{u}, \boldsymbol{u'}) + d_{L}(\{{p}_{\boldsymbol{u}}\}^{t}, \{{p}_{\boldsymbol{u'}}\}^{t}) \ge (t+1) + t = 2t+1$.
		
		Now consider the exceptional case, i.e., when $d_{L}({p}_{\boldsymbol{u}}, {p}_{\boldsymbol{u'}}) =1$, $f(\boldsymbol{u}) \bmod q \in \{0,1,...,(\frac{q}{2}-1)\}$ and $f(\boldsymbol{u'}) \bmod q \in \{\frac{q}{2},(\frac{q}{2}+1),...,(q-1)\}$, and $f(\boldsymbol{u}) -f(\boldsymbol{u'}) = 1 $. For any $\mathbf{u},\mathbf{u'}$ with \( f(\mathbf{u}) \neq f(\mathbf{u'}) \), \(d_{L}(\boldsymbol{u}, \boldsymbol{u'}) \ge 1 \). Therefore, by the encoding given in Construction \ref{Cons1}, we have, $d_{L}(\mathrm{Enc}(\boldsymbol{u}), \mathrm{Enc}(\boldsymbol{u'})) = d_{L}(\boldsymbol{u}, \boldsymbol{u'}) + d_{L}(\{{p}_{\boldsymbol{u}}\}^{t}, \{{p}_{\boldsymbol{u'}}\}^{t}) + d_{L}(\{{p'}_{\boldsymbol{u}}\}, \{{p'}_{\boldsymbol{u'}}\}) \ge 1 + t + d_{L}(0,\frac{q}{2}) = 1 + t + \frac{q}{2} \ge 1+t + t =2t+1$. 
		
		\noindent $\bullet$ Case $3$:  $d_{L}({p}_{\boldsymbol{u}}, {p}_{\boldsymbol{u'}}) \ge 2$. The argument is same as that of Case $3$ for even values of $q$ in the proof of Lemma \ref{Lem_Cons1_FCLC}. \hfill $\blacksquare$.
		
\section{Proof of Lemma \ref{Lem_Cons2_FCLC}}\label{appendix:Lem_Cons2_FCLC}		
	
	For the modular sum function, we first show that \(d_{L}(\boldsymbol{u}, \boldsymbol{u'}) \ge d_{L}(f(\mathbf{u}), f(\mathbf{u'}) ) \). 
	
	Since \( f(\mathbf{u}) \neq f(\mathbf{u'}) \), WLOG assume \( f(\mathbf{u}) > f(\mathbf{u'}) \). Let $\boldsymbol{u}=(u_1,u_2,...,u_k)$ and $\boldsymbol{u'}=(u_1',u_2',...,u_k')$. We have,
	$d_{L}\mathbf{(u,u')} = \sum_{i=1}^{k} w_{L} ((u_i - u'_i) \bmod q)$. And,
	\allowdisplaybreaks
	\begin{align*}
		 d_{L}(f(\boldsymbol{u}), f(\boldsymbol{u'})) & = w_L((f(\mathbf{u}) - f(\mathbf{u'})) \\& = w_L((\sum_{i=1}^{k} {u}_{i}) \bmod q - (\sum_{i=1}^{k} {u'}_{i}) \bmod q)  \\& = w_L( (\sum_{i=1}^{k}(u_i - u'_i) \bmod q) \bmod q) \\& \le \sum_{i=1}^{k} w_L((u_i - u'_i) \bmod q), \text{since Lee wight satisfies triangular inequality \cite{RS}} \\& = d_{L}\mathbf{(u,u')}.
	\end{align*} 
	
	Now the proof of  Lemma \ref{Lem_Cons2_FCLC} is done for odd and even $q$ separately. \\
	
	\noindent $1.$ Proof for odd values of $q$. \\
	Since $f(\mathbf{u}) \in \{0,1,2,...,q-1\}$ and $q$ is odd, the mapping $f(\mathbf{u}) \rightarrow p_{\mathbf{u}} =(2f(\mathbf{u})) \bmod q $ in (\ref{eq6}) is a bijection on $\{0,1,2,...,q-1\}$. Therefore, when \( f(\mathbf{u}) \neq f(\mathbf{u'}) \), we have $ p_{\mathbf{u}} \neq p_{\mathbf{u'}} $. Consider the following cases. \\
	\noindent $\bullet$ Case $1$: $d_{L}({p}_{\boldsymbol{u}}, {p}_{\boldsymbol{u'}}) =1$. Since $f(\mathbf{u}) \in \{0,1,2,...,q-1\}$, we have $f(\mathbf{u}) \bmod q = f(\mathbf{u})$. Therefore, (\ref{eq6}) in Construction \ref{Cons2_MS} is same as (\ref{Eq_Cons1}) in Construction \ref{Cons1}. As shown in the proof of Lemma \ref{Lem_Cons1_FCLC}, for odd $q$ we have, $d_{L}({p}_{\boldsymbol{u}}, {p}_{\boldsymbol{u'}}) =1$ implies either $f(\mathbf{u}) - f(\mathbf{u'}) \equiv \frac{q-1}{2} \pmod q$ or $f(\mathbf{u}) - f(\mathbf{u'}) \equiv \frac{q+1}{2} \pmod q$. The Lee distance between any two elements $x$ and $y$ of the integers $ \bmod \text{ } q$ is the Lee weight of $(x-y) \bmod q$ \cite{CW}. Therefore, we have $d_{L}(f(\boldsymbol{u}), f(\boldsymbol{u'}))=w_L((f(\mathbf{u}) - f(\mathbf{u'})) \bmod q)= w_L(\frac{q-1}{2})=w_L(\frac{q+1}{2})=\frac{q-1}{2}$. Since \(d_{L}(\boldsymbol{u}, \boldsymbol{u'}) \ge d_{L}(f(\mathbf{u}), f(\mathbf{u'}) ) \), we have \(d_{L}(\boldsymbol{u}, \boldsymbol{u'}) \ge \frac{q-1}{2} \). Therefore, $d_{L}(\mathrm{Enc}(\boldsymbol{u}), \mathrm{Enc}(\boldsymbol{u'})) = d_{L}(\boldsymbol{u}, \boldsymbol{u'}) + d_{L}(\{{p}_{\boldsymbol{u}}\}^{t}, \{{p}_{\boldsymbol{u'}}\}^{t}) \ge \frac{q-1}{2} + t \ge \frac{2t+3-1}{2} +t = 2t+1$.
	
	\noindent $\bullet$ Case $2$:  $d_{L}({p}_{\boldsymbol{u}}, {p}_{\boldsymbol{u'}}) \ge 2$. For any $\mathbf{u},\mathbf{u'}$ with \( f(\mathbf{u}) \neq f(\mathbf{u'}) \), \(d_{L}(\boldsymbol{u}, \boldsymbol{u'}) \ge 1 \). Therefore, $d_{L}(\mathrm{Enc}(\boldsymbol{u}),\mathrm{Enc}(\boldsymbol{u'})) = d_{L}(\boldsymbol{u}, \boldsymbol{u'}) + d_{L}(\{{p}_{\boldsymbol{u}}\}^{t}, \{{p}_{\boldsymbol{u'}}\}^{t}) \ge 1 + 2t = 2t+1$. \\
	
	\noindent $2.$ Proof for even values of $q$. \\
	From (\ref{eq6}), it is clear that, when $f(\mathbf{u}) \in \{0,1,2,...,\frac{q}{2}-1\}$, we have $ p_{\mathbf{u}} \in \{0,2,4,...,q-4,q-2\}$,  and when $f(\mathbf{u}) \in \{\frac{q}{2},\frac{q}{2}+1,\frac{q}{2}+2,...,q-2,q-1\}$, we have $ p_{\mathbf{u}} \in \{1,3,5,...q-3,q-1\}$. Therefore, the mapping $f(\mathbf{u}) \rightarrow p_{\mathbf{u}}$ in (\ref{eq6}) is a bijection on $\{0,1,2,...,q-1\}$. Therefore, when \( f(\mathbf{u}) \neq f(\mathbf{u'}) \), we have $ p_{\mathbf{u}} \neq p_{\mathbf{u'}} $. Consider the following cases. 
	
\noindent $\bullet$ Case $1$: $d_{L}({p}_{\boldsymbol{u}}, {p}_{\boldsymbol{u'}}) =1$ implies either $f(\mathbf{u}) - f(\mathbf{u'}) \equiv \frac{q}{2} \pmod q$ or $f(\mathbf{u}) - f(\mathbf{u'}) \equiv \frac{q}{2}-1 \pmod q$, except in the case where $(f(\boldsymbol{u}),f(\boldsymbol{u'}))=(0,q-1)$ or $(q-1,0)$. Therefore, excluding the exceptional case, we have either $d_{L}(f(\boldsymbol{u}), f(\boldsymbol{u'}))=w_L((f(\mathbf{u}) - f(\mathbf{u'})) \bmod q)= w_L(\frac{q}{2})=\frac{q}{2}$ or $d_{L}(f(\boldsymbol{u}), f(\boldsymbol{u'}))= w_L(\frac{q}{2}-1)=\frac{q}{2}-1$. Since \(d_{L}(\boldsymbol{u}, \boldsymbol{u'}) \ge d_{L}(f(\mathbf{u}), f(\mathbf{u'}) ) \), we have \(d_{L}(\boldsymbol{u}, \boldsymbol{u'}) \ge \frac{q}{2}-1 \). Therefore, $d_{L}(\mathrm{Enc}(\boldsymbol{u}), \mathrm{Enc}(\boldsymbol{u'})) = d_{L}(\boldsymbol{u}, \boldsymbol{u'}) + d_{L}(\{{p}_{\boldsymbol{u}}\}^{t}, \{{p}_{\boldsymbol{u'}}\}^{t}) + d_{L}(\{{p'}_{\boldsymbol{u}}\},  \{{p'}_{\boldsymbol{u'}}\}) \ge d_{L}(\boldsymbol{u}, \boldsymbol{u'}) + d_{L}(\{{p}_{\boldsymbol{u}}\}^{t}, \{{p}_{\boldsymbol{u'}}\}^{t}) \ge \frac{q}{2}-1 + t \ge \frac{2t+3}{2} +t = 2t+\frac{3}{2} > 2t+1$.

Now consider the exceptional case, i.e., when $(f(\boldsymbol{u}),f(\boldsymbol{u'}))=(0,q-1)$ or $(q-1,0)$. In this case, by $(7)$ we have $({p'}_{\boldsymbol{u}},{p'}_{\boldsymbol{u'}})=(0,\frac{q}{2})$ or $(\frac{q}{2},0)$. That is, $d_{L}(\{{p'}_{\boldsymbol{u}}\}, \{{p'}_{\boldsymbol{u'}}\}) =\frac{q}{2}$.
Therefore, by the encoding given in Construction $2$, we have $d_{L}(\mathrm{Enc}(\boldsymbol{u}), \mathrm{Enc}(\boldsymbol{u'})) = d_{L}(\boldsymbol{u}, \boldsymbol{u'}) + d_{L}(\{{p}_{\boldsymbol{u}}\}^{t}, \{{p}_{\boldsymbol{u'}}\}^{t}) + d_{L}(\{{p'}_{\boldsymbol{u}}\}, \{{p'}_{\boldsymbol{u'}}\}) \ge 1 + t + d_{L}(0,\frac{q}{2}) = 1 + t + \frac{q}{2} \ge 1+t + t + \frac{3}{2} > 2t+1$.   
	
	\noindent $\bullet$ Case $2$:  $d_{L}({p}_{\boldsymbol{u}}, {p}_{\boldsymbol{u'}}) \ge 2$. For any $\mathbf{u},\mathbf{u'}$ with \( f(\mathbf{u}) \neq f(\mathbf{u'}) \), \(d_{L}(\boldsymbol{u}, \boldsymbol{u'}) \ge 1 \).  Therefore, by the encoding given in Construction \ref{Cons2_MS}, we have, $d_{L}(\mathrm{Enc}(\boldsymbol{u}), \mathrm{Enc}(\boldsymbol{u'})) = d_{L}(\boldsymbol{u}, \boldsymbol{u'}) + d_{L}(\{{p}_{\boldsymbol{u}}\}^{t}, \{{p}_{\boldsymbol{u'}}\}^{t}) + d_{L}(\{{p'}_{\boldsymbol{u}}\}, \{{p'}_{\boldsymbol{u'}}\}) \ge d_{L}(\boldsymbol{u}, \boldsymbol{u'}) + d_{L}(\{{p}_{\boldsymbol{u}}\}^{t}, \{{p}_{\boldsymbol{u'}}\}^{t}) \ge 1 +2t =2t+1$.

\end{appendices}

\end{document}